\documentclass[11pt]{article}

\usepackage[letterpaper,margin=1.00in]{geometry}
\usepackage{amsmath, amssymb, amsthm, thmtools, amsfonts,thm-restate}
\usepackage{bbm}

\usepackage[utf8]{inputenc}
\usepackage[T1]{fontenc}
\usepackage{lmodern}

\usepackage{cite}
\usepackage{appendix}
\usepackage{graphicx}
\usepackage{color}
\usepackage{epstopdf}
\usepackage{wrapfig}
\usepackage{paralist}
\usepackage{todonotes}

\usepackage{framed}
\usepackage[framemethod=tikz]{mdframed}
\usepackage[bottom]{footmisc}
\usepackage[shortlabels]{enumitem}
\setlist[itemize]{noitemsep,topsep=3pt,parsep=3pt,partopsep=3pt}
\setlist[enumerate]{noitemsep,topsep=3pt,parsep=3pt,partopsep=3pt}
\usepackage[font=small]{caption}
\usepackage{xspace}

\definecolor{darkgreen}{rgb}{0,0.5,0}
\usepackage{hyperref}
\hypersetup{
    unicode=false,          
    colorlinks=true,        
    linkcolor=red,          
    citecolor=darkgreen,        
    filecolor=magenta,      
    urlcolor=cyan           
}

\usepackage{algorithm}
\usepackage{algorithmicx}
\usepackage[noend]{algpseudocode}

\usepackage[capitalize, nameinlink]{cleveref}
\crefname{theorem}{Theorem}{Theorems}
\Crefname{lemma}{Lemma}{Lemmas}
\Crefname{claim}{Claim}{Claims}
\Crefname{observation}{Observation}{Observations}
\Crefname{enumi}{Line}{Lines}

\newtheorem{theorem}{Theorem}[section]
\newtheorem{lemma}[theorem]{Lemma}
\newtheorem{meta-theorem}[theorem]{Meta-Theorem}

\newtheorem{corollary}[theorem]{Corollary}

\newtheorem{observation}[theorem]{Observation}
\newtheorem{definition}[theorem]{Definition}

\algnewcommand\algorithmicswitch{\textbf{switch}}
\algnewcommand\algorithmiccase{\textbf{case}}

\algdef{SE}[SWITCH]{Switch}{EndSwitch}[1]{\algorithmicswitch\ #1\ \algorithmicdo}{\algorithmicend\ \algorithmicswitch}%
\algdef{SE}[CASE]{Case}{EndCase}[1]{\algorithmiccase\ #1}{\algorithmicend\ \algorithmiccase}%
\algtext*{EndSwitch}%
\algtext*{EndCase}%

\newcommand{\eps}{\varepsilon}

\newcommand{\local}{$\mathsf{LOCAL}$\xspace}
\newcommand{\congestClique}{$\mathsf{CONGESTED}$-$\mathsf{CLIQUE}$\xspace}
\newcommand{\CONGESTED}{\congestClique}

\newcommand{\poly}{\operatorname{\text{{\rm poly}}}}
\newcommand{\polylog}{\operatorname{\text{{\rm polylog}}}}

\newcommand{\abs}[1]{\left| #1 \right|}

\newcommand{\prob}[1]{P(#1)}
\newcommand{\Prob}[1]{P\left(#1\right)}
\newcommand{\E}[1]{\mathbb{E}\left[{#1}\right]}

\newcommand{\NAlocal}{N_A^{\rm{local}}}
\newcommand{\Nlocal}{N^{\rm{local}}}
\newcommand{\Blocal}{B^{\rm{local}}}
\newcommand{\NAideal}{N_A^{\rm{central}}}
\newcommand{\machines}{m}

\newcommand{\Mopt}{M^{\star}}
\newcommand{\tC}{\tilde{C}}
\newcommand{\tx}{\tilde{x}}
\newcommand{\tClate}{\tC^{late}}

\newcommand{\cT}{\mathcal{T}}

\newcommand{\cA}{\mathcal{A}}
\newcommand{\hatt}{\hat{t}}
\newcommand{\ty}{\tilde{y}}

\newcommand{\tO}{\tilde{O}}
\newcommand{\yold}{y^{old}}
\newcommand{\ylocal}{\ty}
\newcommand{\yideal}{y}
\newcommand{\yMPC}{y^{MPC}}
\newcommand{\xMPC}{x^{MPC}}
\newcommand{\nlocal}{n^{local}}
\newcommand{\nlate}{n^{late}}
\newcommand{\istar}{i^{\star}}
\newcommand{\tstar}{t^{\star}}
\newcommand{\rb}[1]{\left( #1 \right)}

\newcommand{\spec}{\star}

\newcommand{\different}[1]{{\color{blue} #1}}
\newcommand{\eqdef}{\stackrel{\text{\tiny\rm def}}{=}}

\newcommand{\Local}{\textsc{Central}\xspace}
\newcommand{\LocalRand}{\textsc{Central-Rand}\xspace}
\newcommand{\MPCsimul}{\textsc{MPC-Simulation}\xspace}

\DeclareMathOperator{\diff}{diff}
\DeclareMathOperator{\difflocal}{diff^{local}}

\renewcommand{\paragraph}[1]{\vspace{0.15cm}\noindent {\bf #1}:}

\newcommand{\FullOrShort}{full}

\ifthenelse{\equal{\FullOrShort}{full}}{
	
  \newcommand{\fullOnly}[1]{#1}
  \newcommand{\shortOnly}[1]{}
  
  }{

    \newcommand{\fullOnly}[1]{}
    \newcommand{\IncludePictures}[1]{}
   
  }

\begin{document}

\date{}

\title{Improved Massively Parallel Computation Algorithms for \\ MIS, Matching, and Vertex Cover}

\author{
 Mohsen Ghaffari\\
  \small ETH Zurich \\
  \small ghaffari@inf.ethz.ch
\and
Themis Gouleakis \\
\small MIT \\
\small tgoule@mit.edu
\and
Christian Konrad \\
\small University of Bristol \\
\small christian.konrad@bristol.ac.uk
\and
Slobodan Mitrovi{\'c} \\
\small EPFL \\
\small slobodan.mitrovic@epfl.ch
\and
Ronitt Rubinfeld \\
\small MIT and Tel Aviv University \\
\small ronitt@csail.mit.edu
}

\maketitle

\setcounter{page}{0}
\thispagestyle{empty}

\begin{abstract}
We present $O(\log\log n)$-round algorithms in the Massively Parallel Computation (MPC) model, with $\tilde{O}(n)$ memory per machine, that compute a maximal independent set, a $1+\eps$ approximation of maximum matching, and a $2+\eps$ approximation of minimum vertex cover, for any $n$-vertex graph and any constant $\eps>0$. These improve the state of the art as follows: 
\begin{itemize}
\item Our MIS algorithm leads to a simple $O(\log\log \Delta)$-round MIS algorithm in the \congestClique model of distributed computing, which improves on the $\tilde{O}(\sqrt{\log \Delta})$-round algorithm of Ghaffari [PODC'17]. 
\item Our $O(\log\log n)$-round $(1+\eps)$-approximate maximum matching algorithm simplifies or improves on the following prior work: $O(\log^2\log n)$-round $(1+\eps)$-approximation algorithm of Czumaj et al. [STOC'18] and $O(\log\log n)$-round $(1+\eps)$-approximation algorithm of Assadi et al. [SODA'19]. 
\item Our $O(\log\log n)$-round $(2+\eps)$-approximate minimum vertex cover algorithm improves on an $O(\log\log n)$-round $O(1)$-approximation of Assadi et al. [arXiv'17]. 
\end{itemize} 

\end{abstract}
\newpage

\section{Introduction}
	A growing need to process massive data led to development of a number of frameworks for large-scale computation, such as MapReduce~\cite{dg04}, Hadoop~\cite{White:2012}, Spark~\cite{ZahariaCFSS10}, or Dryad~\cite{Isard:2007}. Thanks to their natural approach to processing massive data, these frameworks have gained great popularity. In this work, we consider the \emph{Massively Parallel Computation} (MPC) model~\cite{KarloffSV10} that is abstracted out of the capabilities of these frameworks.

	In our work, we study some of the most fundamental problems in algorithmic graph theory: maximal independent set (MIS), maximum matching and minimum vertex cover. The study of these problems in the models of parallel computation dates back to PRAM algorithm. A seminal work of Luby~\cite{Luby86} gives a simple randomized algorithm for constructing MIS in $O(\log{n})$ PRAM rounds. When this algorithm is applied to the line graph of input graph $G$, it outputs a maximal matching of $G$, and hence a $2$-approximate maximum matching of $G$. The output maximal matching also provides a $2$-approximate minimum vertex cover. Similar results, also in the context of PRAM algorithms, were obtained in~\cite{AlonBI86, II86, IsraeliS86}. Since then, the aforementioned problems were studied quite extensively in various models of computation. In the context of MPC, we design simple randomized algorithms that construct (approximate) instances for all the three problems.
	
\subsection{The Models}
We consider two closely related models: \emph{Massively Parallel Computation} (MPC), and the \congestClique model of distributed computing. Indeed, we consider it as a conceptual contribution of this paper to (further) exhibit the proximity of these two models. We next review these models.
	
		\subsubsection{The MPC model}
		
		The MPC model was first introduced in~\cite{KarloffSV10} and later refined in~\cite{goodrich2011sorting, BeameKS13, AndoniNOY14}. The computation in this model proceeds in synchronous \emph{rounds} carried out by $\machines$ machines. At the beginning of every round, the data (e.g. vertices and edges) is distributed across the machines. During a round, each machine performs computation locally without communicating to other machines. At the end of the round, the machines exchange messages which are used to guide the computation in the next round. In every round, each machine receives and outputs messages that fit into its local memory.
		
	\paragraph{Space} In this model, each machine has $S$ words of space.
	If $N$ is the total size of the data and each machine has $S$ words of space, the typical settings that are of
interest are when $S$ is sublinear in $N$ and $S \cdot m = \Theta(N)$. That is, the total memory across all the machines suffices to fit all the data, but is not much larger than that. If we are given a graph on $n$ vertices, in our work we consider the regimes in which $S \in \Theta(n / \polylog{n})$ or $S \in \Theta(n)$.
	
	\paragraph{Communication vs. computational complexity}
	Our main focus is the number of rounds required to finish the computation, which is essentially the complexity of the communication needed to solve the problem. Although we do not explicitly state the computational complexity in our results, it will be apparent from the description of our algorithms that the total computation time across all the machines is nearly-linear in the input size.
		
	\subsubsection{\congestClique}
	A second model that we consider is the \congestClique model of distributed computing, which was introduced by Lotker, Pavlov, Patt-Shamir, and Peleg\cite{lotker2003mst} and has been studied extensively since then, see e.g.,\cite{Patt-Shamir2011sorting, dolev2012tri, berns2012super, 
  lenzen2013route, Drucker:congestedK, Danupon-paths, hegeman2014near,
  hegeman2014lessons, Censor-Hillel:Algebraic, Hegeman:MST,
  Becker:2015, Ghaffari-MIS, GhaffariMSTLogStar,
  korhonen2016deterministic, henzinger2016deterministic, censor2017derandomizing, ghaffari2017CCMIS, jurdzinski2018mst}. In this model, we have $n$ players which can communicate in synchronous rounds. In each round, every player can send $O(\log n)$ bits to every other player. Besides this communication restriction, the model does not limit the players, e.g., they can use large space and arbitrary computations; though, in our algorithms, both of these will be small. Furthermore, in studying graph problems in this model, the standard setting is that we have an $n$-vertex graph $G=(V, E)$, and each player is associated with one vertex of this graph. Initially, each player knows only the edges incident on its own vertex. At the end, each player should know the part of the output related to its own vertex, e.g., whether its vertex is in the computed maximal independent set or not, or whether some of its edges is in the matching or not.
	
	We emphasize that \congestClique provides an all-to-all communication model. It is worth contrasting this with the more classical models of distributed computing. For instance, the \local model, first introduced by Linial~\cite{linial1987LOCAL}, allows the players to communicate only along the edges of the graph problem $G$ (with unbounded size messages).
	
	\subsection{Related Work}
		\paragraph{Maximum Matching and Minimum Vertex Cover}
			 If the space per machine is $O(n^{1 + \delta})$, for any $\delta > 0$, Lattanzi et al.~\cite{LattanziMSV11} show how to construct a maximal matching, and hence a $2$-approximate minimum vertex cover, in $O(1 / \delta)$ MPC rounds. Furthermore, in case the machine-space is $\Theta(n)$, their algorithm requires $O(\log{n})$ many rounds to output a maximal matching. In their work, they apply \emph{filtering} techniques to gradually sparsify the graph. Ahn and Guha~\cite{AhnG15} provide a method for constructing a $(1 + \eps)$-approximation of weighted maximum matching in $O(1 / (\delta \eps))$ rounds while, similarly to~\cite{LattanziMSV11}, requiring that the space per machine is $O(n^{1 + \delta})$.
			
			If the space per machine is $\tO(n \sqrt{n})$, Assadi and Khanna~\cite{AssadiK17} show how construct an $O(1)$-approximate maximum matching and an $O(\log{n})$-approximate minimum vertex cover in two rounds. Their approach is based on designing randomized composable coresets.
			
			Recently, Czumaj et al.~\cite{czumaj2017round} designed an algorithm for constructing a $(1+\eps)$-approximate maximum matching in $O((\log \log{n})^2)$ MPC rounds of computation and $O(n / \polylog{n})$ memory per machine. To obtain this result, they start from a variant of a PRAM algorithm that requires $O(\log{n})$ parallel iterations, and showed how to compress many of those iterations (on average, $O(\log{n} / (\log \log{n})^2)$ many of them) into $O(1)$ MPC rounds. Their result does not transfer to an algorithm for computing $O(1)$-approximate minimum vertex cover.
			
			Building on~\cite{czumaj2017round} and~\cite{AssadiK17}, Assadi~\cite{Assadi17} shows how to produce an $O(\log{n})$-approximate minimum vertex cover in $O(\log \log{n})$ MPC rounds when the space per machine is $O(n / \polylog{n})$. The work by Assadi et al.~\cite{ABBMS17} also addresses these two problems, and provides a way to construct a $(1 + \eps)$-approximate maximum matching and an $O(1)$-approximate minimum vertex cover in $O(\log \log{n})$ rounds when the space per machine is $\tO(n)$. Their result builds on techniques originally developed in the context of dynamic matching algorithms and composable coresets.

		\paragraph{Maximal Independent Set} Maximal independent set has been central in the study of graph algorithms in both the parallel and the distributed models. The seminal work of Luby~\cite{Luby86} and Alon, Babai, and Itai~\cite{AlonBI86} provide $O(\log n)$-round parallel and distributed algorithms for constructing MIS. The distributed complexity in the \local model was first improved by Barenboim et al.\cite{barenboim2012locality} and consequently by Ghaffari~\cite{Ghaffari-MIS}, which led to the current best round complexity of $O(\log \Delta) + 2^{O(\sqrt{\log\log n})}$. In the \congestClique{} model of distributed computing, Ghaffari~\cite{ghaffari2017CCMIS} gave another algorithm which computes an MIS in $\tilde{O}(\sqrt{\log \Delta})$ rounds. A deterministic $O(\log n\log \Delta)$-round \congestClique{} algorithm was given by Censor-Hillel et al.~\cite{censor2017derandomizing}.
		
		It is also worth referring to the literature on one particular MIS algorithm, known as the \emph{randomized greedy MIS}, which is relevant to what we do for MIS. In this algorithm, we permute the vertices uniformly at random and then add them to the MIS greedily. Blelloch et al.~\cite{blelloch2012greedy} showed that one can implement this algorithm in $O(\log^2 n)$ parallel/distributed rounds, and recently Fischer and Noever~\cite{fischer2018tight} improved that to a tight bound of $\Theta(\log n)$. We will show a $O(\log\log \Delta)$-round simulation of the {randomized greedy MIS} algorithm in the MPC and the \congestClique model.

	\subsection{Our Contributions}
	\label{sec:our-contributions}
		
		As our first result, in \cref{sec:MIS} we present an algorithm for constructing MIS.
		\begin{restatable}{theorem}{theoremMIS}\label{thm:MIS} 
			There is an algorithm that with high probability computes an MIS in $O(\log\log \Delta)$ rounds of the MPC model, with $\tilde{O}(n)$-bits of  memory per machine. Moreover, the same algorithm can be adapted to compute an MIS in $O(\log\log \Delta)$ rounds of the \congestClique model.
		\end{restatable}

		As our second result, in \cref{sec:matching-VC}, we first design an algorithm that returns a $(2+\eps)$-approximate fractional maximum matching and a $(2+\eps)$-approximate integral minimum vertex cover in $O(\log \log{n})$ MPC rounds. Then, in \cref{sec:randomized-rounding}, we show how to round this fractional matching to a $(2 + \eps)$-approximate integral maximum matching.
		In comparison to previous work: our result has somewhat better round-complexity than~\cite{czumaj2017round}, provides a stronger approximation guarantee than~\cite{ABBMS17}, and appears to be simpler than both. After applying vertex-based random partitioning (that was proposed in this context in~\cite{czumaj2017round}), the algorithm repeats only a couple of simple steps to perform all its decisions.
		\begin{restatable}{theorem}{theoremMatchingVC}\label{thm:matching-VC} 
			There is an algorithm that with high probability computes a $(2 + \eps)$-approximate integral maximum matching and a $(2 + \eps)$-approximate integral minimum vertex cover in $O(\log\log n)$ rounds of the MPC model, with $\tilde{O}(n)$-bits of  memory per machine.
		\end{restatable}
		Following similar observations as Assadi et al.~\cite{ABBMS17}, it is possible to apply the techniques of~\cite{McGregor05} on \cref{thm:matching-VC} to obtain the following result.
		\begin{corollary}
			There exists an algorithm that with high probability constructs a $(1 + \eps)$-approximate integral maximum matching in $O(\log \log{n}) \cdot (1/\eps)^{O(1 / \eps)}$ MPC rounds, with $\tilde{O}(n)$-bits of memory per machine.
		\end{corollary}
		As noted by Czumaj et al.~\cite{czumaj2017round}, the result of Lotker et al.~\cite{lotkerMatching} can be used to obtain the following result.
		\begin{corollary}
			There exists an algorithm that outputs a $(2+\eps)$-approximation to maximum weighted matching in $O(\log \log{n} \cdot (1/ \eps))$ MPC rounds and $\tO(n)$-bits of memory per machine.
		\end{corollary}

		For the sake of clarity, we present our algorithms for the case in which each machine has $\tO(n)$-bits of memory (or $O(n)$ words of memory). However, similarly to~\cite{czumaj2017round}, our algorithm for matching and vertex cover can be adjusted to still run in $O(\log \log{n})$ MPC rounds even when the memory per machine is $O(n / \polylog{n})$.
		
		\subsection{Our Techniques}
			\paragraph{Maximal independent set}
				Our MPC algorithm for MIS is based on the randomized greedy MIS algorithm. We show how to efficiently implement this algorithm in only $O(\log \log {n})$ MPC and \congestClique rounds.
				
			\paragraph{Maximum matching and vertex cover}
				In \cref{sec:basic-LP-alg}, we start from a sequential algorithm that outputs a $(2 + \eps)$-approximate fractional maximum matching and a $(2 + \eps)$-approximate integral minimum vertex cover. The algorithm maintains edge-weights. Initially, every edge-weight is set to $1/n$. Then, gradually, at each iteration the edge-weights are simultaneously increased by a multiplicative factor of $1 / (1 - \eps)$. Each vertex whose sum of the incident edges becomes $1 - 2\eps$ or larger is frozen, and its incident edges do not change their weights afterward. The vertices that are frozen in this process constitute the desired vertex cover. It is not hard to see that after $O(\log{n} / \eps)$ iterations every edge will be incident to at least one frozen vertex, and at this point the algorithm terminates.

				In \cref{sec:actual-simulation-in-MPC}, we show how to simulate this sequential algorithm in the MPC model, by on average simulating $\Theta(\log{n} / \log \log{n})$ iterations in $O(1)$ MPC rounds. As the first step, motivated by~\cite{czumaj2017round}, we apply vertex-based sampling. Namely, the vertex-set is randomly partitioned across the machines into disjoint sets, and each machine considers only the induced graph on its local copy of vertices. 
				Then, during each MPC round, every machine simulates several iterations of the sequential algorithm on its local subgraph. During this simulation, each machine estimates weights of the vertices that it maintains locally in order to decide which vertices should be frozen. However, even if the estimates are sharp, only a slight error could potentially cause many vertices to largely deviate from their true behavior. To alleviate this issue, instead of having a fixed threshold $1 - 2\eps$, for each vertex and in every iteration we choose a random threshold from the interval $[1 - 4\eps, 1- 2 \eps]$. For most vertices, this prevents slight errors in estimates from having large effects. Then, vertices are frozen only if their estimated weight is above their randomly chosen threshold. Intuitively, this significantly reduces the chance of these decisions (on whether to freeze a vertex or not) deviating from the true ones

				As our final component, in \cref{sec:randomized-rounding}, we provide a rounding procedure that for a given fractional matching produces an integral one of size only a constant-factor smaller than the size of the fractional matching. Furthermore, every vertex in that rounding method chooses edges based only on its neighborhood, i.e., makes local decision. Thus it is straightforward to parallelize the rounding procedure.
\section{Preliminaries}
\paragraph{Notation}
	For a graph $G = (V, E)$ and a set $V' \subseteq V$, $G[V']$ denotes the subgraph of $G$ \emph{induced} on the set $V'$, i.e., $G[V'] = (V', E \cap (V' \times V'))$. We use $N(v)$ to refer to the neighborhood of $v$ in $G$.
	Throughout the paper, we use $n := |V|$ to denote the number of vertices in the input graph.




\paragraph{Independent Sets} An {\em independent set} $I \subseteq V$ is a subset of non-adjacent
vertices. An independent set $I$ is {\em maximal} if for every $v \in V \setminus I$, $I \cup \{ v\}$ is not an independent set.

Ghaffari gave the following result that we will reuse in this paper:

\begin{theorem}[Ghaffari \cite{ghaffari2017CCMIS}]
\label{sparse-MIS}
 Let $G$ be an $n$-vertex graph with $\Delta(G) = \poly \log (n)$. Then, there exists a distributed algorithm that runs in the 
 \CONGESTED model and computes an MIS on $G$ in $O(\log \log \Delta)$ rounds.
\end{theorem}

\paragraph{Routing} As a subroutine, our algorithm needs to solve the following simple routing task: Let $u \in V$ be an arbitrary vertex. 
Suppose that every other vertex $v \in V \setminus \{u \}$ holds $0 \le n_{v} \le n$ messages each of size $O(\log n)$ that it wants to 
deliver to $u$. We are guaranteed that $\sum_{v \in V} n_v \le n$. Lenzen proved that in the \CONGESTED model there is a deterministic 
routing scheme that achieves this task in $O(1)$ rounds \cite{lenzen2013route}. In the following, we will refer to this scheme as Lenzen's 
routing scheme.

\paragraph{Relevant Concentration Bounds}
Throughout the paper, we will use the following well-known variants of Chernoff bound.
\begin{theorem}[Chernoff bound]\label{lemma:chernoff}
	Let $X_1, \ldots, X_k$ be independent random variables taking values in $[0, 1]$. Let $X \eqdef \sum_{i = 1}^k X_i$ and $\mu \eqdef \E{X}$. Then,
	\begin{enumerate}[(A)]
		\item\label{item:delta-at-most-1} For any $\delta \in [0, 1]$ it holds $\prob{|X - \mu| \ge \delta \mu} \le 2 \exp\rb{- \delta^2 \mu / 3}$.
		\item\label{item:delta-at-least-1} For any $\delta \ge 1$ it holds $\prob{X \ge (1 + \delta) \mu} \le \exp\rb{- \delta \mu / 3}$.
	\end{enumerate}
\end{theorem}

\section{Maximal Independent Set}
\label{sec:MIS}
The \textsc{Greedy} algorithm for maximal independent set processes the vertices of an input graph in arbitrary
order. The algorithm adds the current vertex under consideration to an initially empty independent set $I$ if none 
of its neighbors are already in $I$. 

This algorithm progressively thins out the input graph, and the rate at which the graph loses edges depends 
heavily on the order in which the vertices are considered.  
Consider a sequential random greedy algorithm that ranks/permutes vertices $1$ to $n$ randomly and then greedily adds vertices to the MIS, while walking through this permutation. As it was observed in \cite{acgmw15} in the context of correlation clustering in the streaming model, the number of edges in the residual graph decreases relatively quickly with high probability. 
 In this section, we simulate this algorithm in $O(\log \log \Delta)$ rounds of the MPC and the \congestClique model, thus proving the following result.

\theoremMIS*

\subsection{Randomized Greedy Algorithm for MIS}
Let us first consider a randomized variant of the sequential greedy MIS algorithm described below, that we show how to implement in the \congestClique and the MPC model. We remark that this algorithm has been studied before in the literature of parallel algorithms\cite{fischer2018tight, blelloch2012greedy}. 

\medskip
\smallskip
\begin{minipage}{0.95\linewidth}
\begin{mdframed}[hidealllines=true, backgroundcolor=gray!15]
\vspace{-3pt}
\paragraph{Greedy Randomized Maximal Independent Set} 
\begin{itemize}
\item[-] Initially, choose a permutation $\pi:[n]\rightarrow [n]$ uniformly at random.  
\item[-] Repeat until the next rank is at least $n / \log^{10}{n}$ and the maximum degree is at most $\log^{10}{n}$:
\begin{itemize}
\item[(A)] Mark the vertex $v$ which has the smallest rank among the remaining vertices according to $\pi$, and add $v$ to the MIS.
\item[(B)] Remove all the neighbors of $v$.
\end{itemize}
\item[-] Run $O(\log \log \Delta)$ rounds of the Sparsified MIS Algorithm of~\cite{ghaffari2017CCMIS} in the remaining graph. Remove from the graph the constructed MIS and its neighborhood.
\item[-] Deliver the remaining graph on a single machine and find its MIS.
\item[-] At the end, output the constructed MIS sets.
\end{itemize}
\end{mdframed}
\end{minipage}
\medskip

\subsection{Simulation in $O(\log\log \Delta)$ rounds of MPC and \congestClique}
\label{sec:MIS-simulation}

\paragraph{Simulation in the MPC model} We now explain how to simulate the above algorithm in the MPC model with $O(n\log n)$-bits of memory per machine, and also in the \congestClique model. In each iteration, we take an induced subgraph of $G$ that is guaranteed to have $\tilde{O}(n)$ edges and simulate the above algorithm on that graph. We show that the total number of edges drops fast enough, so that $O(\log\log \Delta)$ rounds will suffice. 
More concretely, we first consider the subgraph induced by vertices with ranks $1$ to $n/{\Delta}^\alpha$, for $\alpha=3/4$. This subgraph has $O(n)$ edges, with high probability. So we can deliver it to one machine, and have it simulate the algorithm up to this rank. Now, this machine sends the resulting MIS to all other machines. Then, each machine removes its vertices that are in MIS or neighboring MIS. In the second phase, we take the subgraph induced by remaining vertices with ranks $n/\Delta^\alpha$ to $n/\Delta^{\alpha^2}$. Again, we can see that this subgraph has $O(n)$ edges (a proof is given below), so we can simulate it in $O(1)$ rounds. 
More generally, in the $i$-th iteration, we will go up to rank $n/\Delta^{\alpha^i}$. Once the next rank becomes $n / \log^{10}{n}$, which as we show happens after $O(\log \log {\Delta})$ rounds, the maximum degree of the graph is some value $\Delta'\leq O(\log^{11}{n})$ (see \cref{lemma:max-degree-whp}). Note that clearly also $\Delta'\leq \Delta$. At that point, we apply the MIS Algorithm of~\cite{ghaffari2017CCMIS} for sparse graphs to the remaining graph. This algorithm is applicable whenever the maximum degree is at most $2^{O(\sqrt{\log n})}$ (see Theorem~1.1 of~\cite{ghaffari2017CCMIS}). After $O(\log \log \Delta')$ rounds, w.h.p., that algorithm finds an MIS which after removed along with its neighborhood results in the graph having $O(n)$ edges. Now we deliver the whole remaining graph to one machine where it is processed in a single MPC round. 


We note that the Algorithm of~\cite{ghaffari2017CCMIS} performs only simple local decisions with low communication, and hence every iteration of the algorithm can be implemented in $O(1)$ MPC rounds, with $\tO(n)$ memory per machine, by using standard techniques.

\medskip
\paragraph{Simulation in \congestClique} We now argue that each iteration can be implemented in $O(1)$ rounds of \congestClique. To simulate the first step of the algorithm, all vertices agree on a uniform random order as follows: the vertex with the smallest ID choses a uniform random order 
locally and informs all other vertices about their positions within the order. Then, all vertices broadcast their positions to all
other vertices. As a result, all vertices know the entire order. Also, in each iteration, we make all vertices with permutation rank in the selected range send their edges to the leader vertex. Here, the leader is an arbitrarily chosen vertex, e.g., the one with the minimum identifier. As we show below, the number of these edges per iteration is $O(n)$ with high probability, and thus we can deliver all the messages to the leader in $O(1)$ rounds using Lenzen's routing method\cite{lenzen2013route}. Then, the leader can compute the MIS among the vertices with ranks in the selected range. It then reports the result to all the vertices in a single round, by telling each vertex whether it is in the computed independent set or not. A single round of computation, in which the vertices in the independent set report to all their neighbors, is then used to remove all the vertices that have a neighbor in the independent set (or are in the set). After these steps, the algorithm proceeds to the next iteration.

Regarding the round-complexity of the algorithm once the rank becomes $n / \log^{10}{n}$: The work~\cite{ghaffari2017CCMIS} already provides a way to solve MIS in $O(\log \log{\Delta'})$ \congestClique rounds for any $\Delta'=2^{O(\sqrt{\log n})}$. Here, $\Delta'$ is the maximum degree of the graph remained after processing the vertices up to rank $n / \log^{10}{n}$, and, as we show by \cref{lemma:max-degree-whp}, that $\Delta'\leq \polylog{n} \ll 2^{O(\sqrt{\log n})}$. Hence, the overall round complexity is again $O(\log\log \Delta)$ rounds. 

\subsection{Analysis}
Since by the $i$-th iteration the algorithm has processed the ranks up to $n/\Delta^{\alpha^i}$, the rank $n / \log^{10}{n}$ is processed within $O(\log\log \Delta)$ iterations. In the proof of \Cref{thm:MIS} presented below, we prove that with high probability the number of edges sent to one machine per phase is $O(n)$. Before that, we present a lemma that will aid in bounding the degrees and the number of edges in our analysis. A variant of this lemma was proved in \cite{acgmw15}.

\begin{lemma}\label{lemma:max-degree-whp}
	Suppose that we have simulated the algorithm up to rank $r$. Let $G_r$ be the remaining graph. Then, the maximum degree in $G_r$ is $O(n \log n/r)$ with high probability.
\end{lemma}
\begin{proof}
	We first upper-bound the probability that $G_r$ contains a vertex of degree at least $d$. Then, we conclude that the degree of every vertex in $G_r$ is $O(n \log n/r)$ with high probability.

		Consider a vertex $v$ whose degree is still $d$. When the sequential algorithm considers one more vertex, which is selected by choosing a random vertex among the remaining vertices, then vertex $v$ or one of its neighbors gets chosen with probability at least $d/n$. If that happens, then $v$ is removed. The probability that this does not happen throughout ranks $1$ to $r$ is at most $(1-d/n)^{r} \le \exp(- r d / n)$. Now, the probability that a vertex in $G_r$ has degree more than $20n\log n/r$ is at most $1/n^{5}$, which implies that, the maximum degree of $G_r$ is at most $20 n \log{n} / r$ with probability at least $1 - n^{-4}$. 
\end{proof}

We are now ready to prove the main theorem of this section.
\begin{proof} [Proof of \Cref{thm:MIS}]
We first argue about the MPC round-complexity of the algorithm, and then show that it requires $\tO(n)$ memory.

\paragraph{Round complexity}
	Recall that the algorithm considers ranks of the form $r_i := n / \Delta^{\alpha^i}$, until the rank becomes $n/\log^{10}{n}$ or greater. When that occurs, it applies other algorithms for $O(\log \log \Delta)$ iterations, as described in \cref{sec:MIS-simulation}. Hence, the algorithm runs for at most $\istar + \log \log{\Delta}$ iterations, where $\istar$ is the smallest integer such that rank $r_{\istar} := n / \Delta^{\alpha^{\istar}} \ge n / \log^{10}{n}$. A simple calculation gives $\istar \le \log_{4/3} \log \Delta$, for $\alpha = 3/4$. Furthermore, every iteration can be implemented in $O(1)$ rounds as discussed above.
	
\paragraph{Memory requirement}
	We first discuss the memory required to implement the process until the rank becomes $O(n / \log^{10}{n})$. By \cref{lemma:max-degree-whp} we have that after the graph up to rank $r_i$ is simulated, the maximum degree in the remaining graph is $O(n \log{n} / r_i)$ w.h.p. Observe that it also trivially holds in the first iteration, i.e. the initial graph has maximum degree $O(n)$. Let $G_i$ be the graph induced by the ranks between $r_i$ and $r_{i + 1}$. Then, a neighbor $u$ of vertex $v$ appears in $G_i$ with probability $(r_{i + 1} - r_i) / (n - r_i) \le r_{i + 1} / n$. Hence, the expected degree of every vertex in this graph is at most
	\[
		\mu := \Theta(n \log{n} / r_i \cdot r_{i + 1} / n) = \Theta\rb{\Delta^{(1 - \alpha) \alpha^i} \log n}.
	\]
	Since $\mu \ge \log{n}$, by the Chernoff bound (\cref{lemma:chernoff}) we have that every vertex in $G_i$ has degree $O(\mu)$ w.h.p. Now, since there are $O(r_{i + 1})$ vertices in $G_i$, we have that $G_i$ contains
	\begin{equation}\label{eq:G-has-edges}
		O\rb{r_{i + 1} \Delta^{(1 - \alpha) \alpha^i} \log n} = O\rb{n \Delta^{- \alpha^i / 2} \log{n}}
	\end{equation}
	many edges w.h.p., where we used that $\alpha = 3/4$. Recall that the algorithm iterates over the ranks until the maximum degree becomes less than $\log^{10}{n}$. Also, $\Theta(n \log{n} / r_i)$ upper-bounds the maximum degree (see \cref{lemma:max-degree-whp}). Hence, we have
	\[
		\Theta(n \log{n} / r_i) \ge \log^{10}{n} \; \implies \; \Delta^{\alpha^i} \ge \Omega\rb{\log^{9}{n}}.
	\]
	Combining the last implication with \cref{eq:G-has-edges} provides that $G_i$ contains $O(n)$ edges w.h.p.
	
	After the rank becomes $n / \log^{10}{n}$ or greater, we run the \congestClique algorithm of~\cite{ghaffari2017CCMIS} for $O(\log \log \Delta)$ iterations. Since that algorithm performs only simple local decisions with low communication, every iteration of the algorithm can be implemented in $O(1)$ MPC rounds, with $\tO(n)$ memory per machine, by using standard techniques. 
Finally, using \cref{sparse-MIS}, we conclude that the MIS will be computed after $O(\log \log \Delta)$ rounds in the MPC or the \congestClique model.

\end{proof}

\section{Matching and Vertex Cover, Simple Approximations}
\label{sec:matching-VC}
In this section, we describe a simple algorithm that leads to a fractional matching of weight within a $(2+\eps)$-factor of (integral) maximum matching and, the same algorithm, leads to a $2+\eps$ approximation of minimum vertex cover, for any small constant $\eps>0$. In the next section (\cref{sec:randomized-rounding}), we explain how to obtain an integral $(2+\eps)$-approximate maximum matching from the described fractional one. That result, along with standard techniques underlined in \cref{sec:our-contributions}, provides $(1+\eps)$-approximation of maximum matching.

In \cref{sec:basic-LP-alg}, we first present the advertised algorithm that runs in $O(\log{n})$ rounds. Then, in \cref{sec:fractional-LP-in-MPC} and \cref{sec:actual-simulation-in-MPC}, we explain how to simulate this algorithm in $O(\log\log n)$ rounds of the MPC model. In \cref{sec:analysis-matching} we provide the analysis of this simulation. 

\subsection{Basic $O(\log n)$-iteration Centralized Algorithm}
\label{sec:basic-LP-alg}
We now provide a simple centralized algorithm for obtaining the described fractional matching and minimum vertex cover. We refer to this algorithm as $\Local$.

\smallskip
\begin{minipage}{0.95\linewidth}
\begin{mdframed}[hidealllines=true, backgroundcolor=gray!15]
\vspace{-3pt}
\paragraph{$\Local$: Centralized $O(\log n)$-round Fractional Matching and Vertex Cover} 
\begin{itemize}
\item[-] Initially, for each edge $e \in E$, set $x_e=1/n$. 
\item[-] Then, until each edge is frozen, in iteration $t$: 
\begin{itemize}
\item[(A)] Freeze each vertex $v$ for which $y_v=\sum_{e \ni v}x_e \geq 1 - 2 \eps$ and freeze all its edges.
\item[(B)] For each active edge, set $x_e\gets x_e / (1-\eps)$. 
\end{itemize}
\item[-] At the end, once all edges are frozen, output the set of values $x_e$ as a fractional matching and the set of frozen vertices as a vertex cover.
\end{itemize}
\end{mdframed}
\end{minipage}

\begin{lemma}\label{lemma:ideal-alg-guarantee}
For any constants $\eps$ such that $0 < \eps \le 1/10$, the algorithm $\Local$ terminates after $O(\log n)$ iterations, at which point all edges are frozen. Moreover, we have two properties: 
\begin{itemize}
\item[(A)] The set of frozen vertices---i.e., those $v$ for which $y_{v, t}=\sum_{e \ni v}x_e \geq 1-2\eps$---is a vertex cover that has size within a $(2+5 \eps)$ factor of the minimum vertex cover.
\item[(B)] $\sum_{e\in E} x_e \geq |\Mopt|/(2+5\eps)$, that is, the computed fractional matching has size within $(2+5\eps)$-factor of the maximum matching 
\end{itemize}
\end{lemma}
\begin{proof}
We first prove the claim about vertex cover, and then about maximum matching.

\paragraph{Vertex cover}
	Let $C$ be the vertex cover obtained by the algorithm. Every vertex added to $C$ has weight at least $1 - 2 \eps$. Furthermore, an edge can be incident to at most $2$ vertices of $C$. Let $W_M$ be the weight of the fractional matching the algorithm constructs. Then, we have $|C| \le 2 W_M / (1 - 2\eps) \le 2 (1 + 5 \eps) W_M$, for $\eps \le 1/10$. Note that the algorithm ensures that at every step $y_v \le 1$. Hence, from strong duality we have that the weight of fractional minimum vertex cover is at least $W_M$. Therefore, the minimum (integral) vertex cover has size at least $W_M$ as well. This now implies that $|C|$ is a $2 (1 + 5\eps)$-approximate minimum vertex cover.

\paragraph{Maximum matching}
	Let $W_M^{\star}$ be the weight of a fractional maximum matching. Then, it holds $|\Mopt| \le W_M^{\star} \le |C|$. From our analysis above and the last chain of inequalities we have $W_M \ge |\Mopt| / (2 (1 + 5\eps))$.
\end{proof}

\subsection{An Attempt for Simulation in $O(\log \log n)$ rounds of MPC}
\label{sec:fractional-LP-in-MPC}

\paragraph{An Idealized MPC Simulation} Next, we describe an attempt toward simulating the algorithm $\Local$ in the MPC model. Once we describe this, we will point out some shortcomings and then explain how one can adjust the algorithm to address these shortcomings.
  
The algorithm starts with every vertex and every edge being active. If not active, an edge/vertex is frozen. Throughout the algorithm, the minimum active fractional edge value increases and consequently, the degree of each vertex with respect to active edges decreases gradually. We break the simulation into phases, where the $i^{th}$ phase ensures to simulate enough of the algorithm until the minimum active fractional edge value is $1/\Delta^{-(0.9)^i}$, which implies that the active degree is at most $\Delta^{(0.9)^{i}}$. Hence, we finish within $O(\log\log n)$ phases. \textbf{Remark:} In our final implementation, the number of iteration one phase simulates is slightly different than presented here. However, that final implementation, that we precisely define in the sequel, follows the exact same behavior as presented here.

Let us focus on one phase. Suppose that $G'$ is the remaining graph on the active edges, the minimum active fractional edge value is $1/d$, and thus $G'$ has degree at most $d$. In this phase, we simulate the algorithm until the minimum active fractional edge value reaches $1/d^{0.9}$, which implies that the active degree is at most $d^{0.9}$.

We randomly partition the vertex-set of $G'$, that consists only of active edges, among $\machines = \sqrt{d}$ machines; let $G'_i$ be the graph given to machine $i$. In this way, each machine receives $O(n)$ edges w.h.p. Machine $i$ for the next $\log_{1/(1-\eps)} d/10$ rounds simulates the basic algorithm on $G'_i$. For that, in each round the machine which received a vertex $v$  estimates $y_v= \sum_{e\ni v } x_e$
by $\ylocal_v$ defined as
	\[
		\ylocal_v = \machines \cdot \sum_{e \ni v;\, e\in G'_i} x_e + \sum_{e \ni v; \, e\in G\setminus G'} x_e.
	\]
	That is, $\ylocal_v$ is the summation of edge-values of $G'$-edges incident on $v$ whose other endpoint is in the same machine, multiplied by $\machines$ (to normalize for the partitioning), plus the value of all edges remaining from $G\setminus G'$, i.e., edges that were frozen before this phase. In each round and for every vertex $v$, if $\ty_v \geq 1-2 \eps$, then the machine freezes $v$ and the edges incident to $v$. After this step, for any active edge $e \in G'_i$ the machine sets $x_e\gets x_e \cdot 1/(1-\eps)$. The phase ends after $\log_{1/(1-\eps)} \Delta/10$ rounds. At the end, the round in which different vertices were frozen determines when the corresponding edges got frozen (if they did). So, it suffices to spread the information about the frozen vertices and the related timing to deduce the edge-values of all edges. Since per iteration each active edge increases by a factor of $1/(1-\eps)$, after $\log_{1/(1-\eps)} \Delta/10$ rounds, the minimum active edge value reaches $1/d^{0.9}$ and we are done with this phase.

\paragraph{The Issue with the Direct Simulation}
Consider first the following wishful-thinking scenario. Assume for a moment that in \emph{every iteration} it holds $\abs{\yideal_v - (1 - 2 \eps)} > \abs{\yideal_v - \ylocal_v}$, that is, $\yideal_v$ and $\ylocal_v$ are "on the same side" of the threshold. Then, the algorithm $\Local$ and the MPC simulation of it make the same decision on whether a vertex $v$ gets frozen or not. Moreover, this happens in every iteration, as can be formalized by a simple induction. This in turn implies that the MPC algorithm performs the exact same computations as the $\Local$ algorithm and thus it provides the same approximation as $\Local$. However, in general case, even if $\yideal_v$ and $\ylocal_v$ are almost equal, e.g., $\abs{\yideal_v - \ylocal_v} \ll \eps$, it might happen that $\yideal_v \ge 1 - 2 \eps$ and $\ylocal_v < 1 - 2 \eps$, resulting in the two algorithms making different decisions with respect to $v$. Furthermore, this situation could occur for many vertices simultaneously, and this deviation of the two algorithms might grow as we go through the round; these complicate the task of analyzing the behavior of the MPC algorithm. 

\paragraph{Random Thresholding to the Rescue}
Observe that  if $\abs{\yideal_v - \ylocal_v}$ is small then there is only a ``small range" of values of $\yideal_v$ around the threshold $1 - 2 \eps$ which could potentially lead to the two algorithms behaving differently with respect to $v$. Motivated by this observation, instead of having one fixed threshold throughout the whole algorithm, in each iteration $t$ and for each vertex $v$ the algorithm will uniformly at random choose a fresh threshold $\cT_{v, t}$ from the interval $[1 - 4\eps, 1 - 2 \eps]$. We call this algorithm $\LocalRand$, and state it below. Then, if $v$ is not frozen until the $t^{th}$ iteration, $v$ gets frozen by $\LocalRand$ if $\yideal_{v, t} \ge \cT_{v, t}$ (and similarly, $v$ get frozen by the MPC simulation if $\ylocal_{v, t} \ge \cT_{v, t}$). In that case, if $\abs{\yideal_v - \ylocal_v} \ll \eps$, then most of the time $\yideal_v$ would be far from the threshold and the two algorithms would behave similarly. We make this intuition formal in the next section by \cref{claim:prob-to-differ}. 

\subsection{Our Actual Simulation in $O(\log \log n)$ rounds of MPC}
\label{sec:actual-simulation-in-MPC}
We now present the modified $\LocalRand$ algorithm with the random thresholding and then discuss how we simulate it in the MPC model. 

\smallskip
\begin{minipage}{0.95\linewidth}
\begin{mdframed}[hidealllines=true, backgroundcolor=gray!15]
\vspace{-3pt}
\paragraph{$\LocalRand$: Centralized $O(\log n)$-round Fractional Matching and Vertex Cover with Random Thresholding} 
\begin{itemize}
\item[-] \different{Each vertex $v$ chooses a list of thresholds $\cT_{v, t}$ such that: the thresholds are chosen independently; each threshold is chosen uniformly at random from $[1 - 4 \eps, 1 - 2 \eps]$.}
\item[-] Initially, for each edge $e \in E$, set $x_e=1/n$. 
\item[-] Then, until each edge is frozen, in iteration $t$: 
\begin{itemize}
\item[(A)] Freeze each vertex $v$ for which $y_{v, t}=\sum_{e \ni v}x_e \geq \different{\cT_{v, t}}$ and freeze all its edges.
\item[(B)] For each active edge, set $x_e\gets x_e / (1-\eps)$. 
\end{itemize}
\item[-] At the end, once all edges are frozen, output the set of values $x_e$ as a fractional matching and the set of frozen vertices as a vertex cover.
\end{itemize}
\end{mdframed}
\end{minipage}
\medskip

\paragraph{Our Actual MPC Simulation}
We now provide an MPC simulation of $\LocalRand$, that we will refer to by $\MPCsimul$, and discuss it below.

Our algorithm begins by selecting a collection of random thresholds $\cT$. In the actual implementation, since these thresholds are chosen independently and each from the same interval, threshold $\cT_{v, t}$ can be sampled when needed ("on the fly"). During the simulation, we maintain a vertex set $V' \subseteq V$ that consists of vertices that we consider for the rest of the simulation. The algorithm defines the initial weight of the edges to be $w_0 = (1 - 2\eps) / n$. Also, it maintains variable $d$ representing the upper-bound on the maximum degree in the remaining graph (in principle, the maximum degree can be smaller than $d$).

$\MPCsimul$ is divided into phases. At the beginning of a phase, we consider a subgraph $G'$ of $G[V']$ that consists only of the active edges. In \cref{lemma:bound-degree-in-G'} we prove that the maximum degree in $G'$ is at most $d$.
Also at the beginning of a phase, the algorithm defines $\yold_v$ (see \cref{line:define-yold}). This is part of the vertex-weight that remains the same throughout the execution of the phase. It corresponds to the sum of weights of the edges incident to $v$ that were 
frozen in prior phases.
Then, the vertex set $V'$ is distributed across $\machines = \sqrt{d}$ machines. Each machines collects the induced graph of $G'$ on the vertex set assigned to it. In \cref{lemma:size-of-induced-graphs} we prove that each of these induced graphs consists of $O(n)$ edges. 

	Each phase executes the steps under \cref{line:phase}, which simulates $I$ iterations of $\LocalRand$. During a phase, we maintain the iteration-counter $t$. The value of $t$ counts all the iterations since the beginning of the algorithm, and not only from the beginning of a phase. After this simulation is over, the weight $\xMPC_e$ of each edge $e$ is properly set/updated. For instance, if $e$ was not assigned to any of the machines (i.e., its endpoints were assigned to distinct machines), then $\xMPC_e$ was not changing during the simulation of $\LocalRand$ in this phase even if both of its endpoints were active. To account for that, at \cref{line:account-for-distinct-machines} the value $\xMPC_e$ is set to $w_0 \tfrac{1}{(1 - \eps)^{t'}}$, where $t'$ is the last iteration when both endpoints of $e$ were active. To implement this step, each vertex will also keep a variable corresponding to the iteration when it was last active.
	
		Every vertex $v$ that has weight more than $1$, i.e., $\yMPC_v > 1$, is along with its incident edges removed from the consideration, e.g., removed from $V'$ at \cref{line:remove-from-V'}, but $v$ is added to the vertex cover that is reported at the end of the algorithm. Note that after the removal of such $v$, the edges incident to it are not considered anymore while computing $\yMPC$ or $\ylocal$. This step ensures that throughout the algorithm the fractional matching on $G[V']$ will be valid. But it also ensures that all the edges that are in $G[V \setminus V']$, in particular those incident to $v$, will be covered by the final vertex cover.

\smallskip
\begin{minipage}{0.95\linewidth}
\begin{mdframed}[hidealllines=true, backgroundcolor=gray!15]
\vspace{-3pt}
\paragraph{$\MPCsimul$: MPC Simulation of algorithm $\LocalRand$}
\begin{enumerate}[(1)]
	\item Each vertex $v$ chooses a list of thresholds $\cT_{v, t}$ such that: the thresholds are chosen independently; each threshold is chosen uniformly at random from $[1 - 4 \eps, 1 - 2 \eps]$.
	\item Init: $V' = V$; \enskip $\forall e \in E$, set $\xMPC_e = w_0 = \tfrac{1 - 2\eps}{n}$; \enskip $d = n$; \enskip $t = 0$.
	\item While $d > \log^{20}{n}$: 
	\begin{enumerate}[(a)]
		\item\label{line:define-G'} Let $G'$ be a graph on $V'$ consisting only of the active edges of $G[V']$.
		\item\label{line:define-yold} For each $v \in V'$, define $\yold_v = \sum_{e \ni v; \, e \in G[V'] \setminus G'} \xMPC_e$.
		\item\label{line:define-init-for-phase} Set: \# machines $\machines = \sqrt{d}$; \enskip \# iterations $I = \tfrac{\log{\machines}}{10 \log{10}}$.
		\item\label{line:define-Vi} Partition $V'$ into $\machines$ sets $V_1, \ldots, V_{\machines}$ by assigning each vertex to a machine independently and uniformly at random.
		\item\label{line:phase} For each $i \in \{1, \ldots, \machines\}$ in parallel execute $I$ iterations
			\begin{enumerate}[(A)]
				\item\label{line:compute-ylocal} For each $v \in V_i$ such that $\ylocal_{v, t} = \machines \cdot \sum_{e \ni v;\, e \in G'[V_i]}\xMPC_e + \yold_v \geq \cT_{v, t}$: freeze $v$ and freeze all its edges.
				\item For each active edge of $G'[V_i]$, set $\xMPC_e \gets \tfrac{\xMPC_e}{1-\eps}$.
				\item\label{line:increase-active-weight} Increment the total iteration count: $t \gets t + 1$.
			\end{enumerate}
		\item\label{line:update-degree} Update $d \gets d (1 - \eps)^I$.
		\item\label{line:account-for-distinct-machines} For every edge $e = \{u, v\}$: set $\xMPC_e = w_0 \tfrac{1}{(1 - \eps)^{t'}}$, where $t'$ is the last iteration in which both $u$ and $v$ were active.
		\item For each $v \in V'$ let $\yMPC_v = \sum_{e \ni v;\, e \in G[V']} \xMPC_e$.
		\item\label{line:remove-from-V'} For each $v \in V'$ such that $\yMPC_v > 1$: remove $v$ from $V'$.
		\item\label{line:remove-more-1-2eps} For each $v \in V'$ such that $\yMPC_v > 1 - 2 \eps$: freeze $v$ and freeze all its edges.
	\end{enumerate}
	\item\label{line:directly-simulate} Directly simulate $\log_{1 / (1 - \eps)} \log^{20}{n}$ iterations of $\LocalRand$.
	\item Output the vector $\xMPC$ as a fractional matching and the set of frozen vertices as a vertex cover.
\end{enumerate}
\end{mdframed}
\end{minipage}
\smallskip

	If some vertex has weight between $1 - 2 \eps$ and $1$, it has sufficiently large fractional weight, so we simply freeze it (\cref{line:remove-more-1-2eps}) before the next phase.

	Once the upper-bound $d$ becomes less than $\log^{20}{n}$, the algorithm exits from the main while loop, and the rest of the iterations needed to simulate $\LocalRand$ are executed one by one. During this part of the simulation, $\MPCsimul$ and $\LocalRand$ behave identically.


\subsection{Analysis}
\label{sec:analysis-matching}
We prove that the set of frozen vertices forms a $2+O(\eps)$ approximation of the minimum vertex cover, and the computed fractional matching is a $2+O(\eps)$ approximation of maximum matching.

\begin{lemma}\label{lemma:main-lemma-matching}
	$\MPCsimul$ with high probability outputs a $(2 + 50 \cdot \eps)$-approximate minimum vertex cover and a fractional matching which is a $(2 + 50 \cdot \eps)$ approximation of maximum matching. Moreover, there is an implementation of $\MPCsimul$ that with high probability has $O(\log \log{n})$ MPC-round complexity and requires $O(n)$ space per machine.
	
	Furthermore, the algorithm outputs fractional matching $x$ and a vertex cover $C$ such that the fractional weight of at least $|C| / 3$ vertices of $C$ is at least $1 - 5 \eps$.
\end{lemma}

\paragraph{Remark}
	For technical reasons and for the sake of clarity of our exposition, in our analysis we assume that $\eps < 1/50$. (If the input $\eps \ge 1/50$, we simply reduce its value and deliver even better approximation than required.) Also, as $\eps$ is a constant, we assume that $\eps > 1 / \log{n}$.

\paragraph{Roadmap} We split the proof of \cref{lemma:main-lemma-matching} into three parts. We start by, in \cref{sec:simple-weight-properties}, showing some properties of the edge-weights and the maximum degree of vertices during the course of $\MPCsimul$. Then, in \cref{sec:memory-requirement} we prove that $O(n)$ space per machine suffices for the execution of $\MPCsimul$, and that the algorithm can be executed in $O(\log \log{n})$ MPC-rounds. Next, in \cref{sec:weight-properties} we relate the vertex-weights in $\MPCsimul$ (i.e., the vectors $\ylocal$ and $\yMPC$) to the corresponding weights in the algorithm $\LocalRand$ (i.e., to the vector $\yideal$). namely, we trace $\abs{\ylocal_v - \yideal_v}$ over the course of one phase, and show that for most of the vertices this difference remains small. We put forth those results in \cref{sec:proof-of-matching-lemma} and prove \cref{lemma:main-lemma-matching}.

\subsubsection{Weight and degree properties}
\label{sec:simple-weight-properties}
	We now state several properties of edge-weights that are easily derived from the algorithm, and provide an upper-bound on the maximal active degree of any vertex. These properties will be used throughout our proofs in the coming sections. 
	
	Define $w_t = w_0 \tfrac{1}{(1 - \eps)^t}$. Observe that at the $t^{th}$ iteration, the weight of all the active edges that are on some of the machines equals $w_t$. Furthermore, if for an edge $e = \{u, v\}$ such that $u$ and $v$ are on different machines, vertices $u$ and $v$ are both active in the $t^{th}$ iteration, then after the phase ends the weight $\xMPC_e$ will be set to at least $w_t$ (see \cref{line:account-for-distinct-machines}). We next state two observations.
	
	\begin{observation}[Degree --- active-weight invariant]\label{obs:degree-active-weight}
		Consider an iteration $t$ at which is updated $d$ at \cref{line:update-degree} of $\MPCsimul$. Then, just after the degree $d$ is updated, it holds $d \cdot w_t = 1 - 2 \eps$.
	\end{observation}
	\begin{proof}
		At the beginning of the algorithm, it holds $w_0 \cdot d = 1 - 2\eps$. Over a phase, weights of the active edges increase by $1 / (1 - \eps)^I$. On the other hand, the degree $d$ decreases by $1 / (1 - \eps)^I$. Hence, their product remains the same after every phase.
	\end{proof}
	
\begin{observation}[Maximum weight of active-edge]\label{obs:max-active-weight}
	The weight of any active edge at the beginning of a phase is $(1 - 2\eps) / \machines^2$, where $\machines$ is the number of machines used in that phase. During that phase, the weight of any edge is at most $1 / \machines^{1.8}$.
\end{observation}
\begin{proof}
	Let $w_{\tstar}$ be the weight of any active edge at the beginning of a phase. As defined at \cref{line:define-init-for-phase}, we have $\machines^2 = d$. From \cref{obs:degree-active-weight} we hence conclude that $w_{\tstar} = (1 - 2\eps) / \machines^2$.
	
	Also at \cref{line:define-init-for-phase}, $I$ is defined to be $(\log{\machines})(10 \log{10}) < (\log{\machines}) / 10$. On the other hand, for at most $I$ iterations the weight of any active edge is increased by at most $1 / (1 - \eps) \le 2$ at \cref{line:increase-active-weight}. Hence
		\[
			w_{\tstar + I} \le 2^I \cdot (1 - 2 \eps) / \machines^2 \le (1 - 2\eps) \machines^{0.2} / \machines^2 \le \machines^{1.8}.
		\]
\end{proof}

\subsubsection{Memory requirement and round complexity}
\label{sec:memory-requirement}
In this section, we first show that $O(n)$ space per machines suffices to store the induced graphs $G'[V_i]$ considered by $\MPCsimul$ (see \cref{lemma:size-of-induced-graphs}). After, in \cref{lemma:number-of-phases}, we upper-bound the number of phases of $\MPCsimul$. At the end of the section, we combine these together to prove the following lemma.
\begin{lemma}\label{lemma:memory-and-round-complexity}
	There is an implementation of $\MPCsimul$ that requires $O(n)$ memory per machine and executes $O(\log \log{n})$ MPC rounds w.h.p.
\end{lemma}

	We start by upper-bounding the number of active edges incident to a vertex of $V'$.
	\begin{lemma}\label{lemma:bound-degree-in-G'}
		Let $V'$, $G'$ and $d$ be as defined in $\MPCsimul$ at the beginning of the same phase. Then, the degree of every vertex in $G'[V']$ is at most $d$.
	\end{lemma}
	\begin{proof}
		In the beginning of the algorithm, we have that $d = n$, and hence the statement holds for the very first phase.
		
		Towards a contradiction, assume that there exists a phase and a vertex $v$ such that the degree of $v$ in $G'[V']$ is more than $d$. Let $d_v$ denote its degree. Let $\tstar$ be the first iteration of that phase. Notice that $w_{\tstar}$ was the weight of active edges at the end of the previous phase. Now by \cref{obs:degree-active-weight} we have
		\[
			w_{\tstar} \cdot d_v > w_{\tstar} \cdot d = 1 - 2\eps.
		\]
		But this now contradicts the step at \cref{line:remove-more-1-2eps} of $\MPCsimul$ after which all the edges incident to $v$ would become frozen.
	\end{proof}
	Now we prove that every induced graph processed on machine has $O(n)$ edges.
	\begin{lemma}[Size of induced graphs]\label{lemma:size-of-induced-graphs}
		Let $G'$ and $V_i$ be as defined at \cref{line:define-G'} and \cref{line:define-Vi} of $\MPCsimul$, respectively. Than, $\left|E\rb{G'[V_i]}\right| \in O(n)$ w.h.p.
	\end{lemma}
	\begin{proof}
		We split the proof into two parts. First, we argue that the size of $V_i$ is $O(n / m)$ w.h.p. After, we argue that the degree of each vertex in $G'[V_i]$ is $O(d / m)$ w.h.p., after which the proof will follow by union bound.
		
		\paragraph{Expected size of $V_i$} Now, $\E{|V_i|} = |V'| / m \le n / m$. Observe that we have $m \le \sqrt{n}$ at any phase, and hence $n / m \ge \sqrt{n}$. Now Chernoff bound (\cref{lemma:chernoff}) implies that inequality
		\begin{equation}\label{eq:Vi-bound}
			|V_i| \le |V'| / m + n / m \in O(n / m)
		\end{equation}
		holds w.h.p.
		
		\paragraph{Degree bound}
			Consider a vertex $v \in V_i$, and let $d_v$ be its degree in $V_i$. \cref{lemma:bound-degree-in-G'} implies $\E{d_v} \le d / m$. By the definition it holds $d / m = m$, and also $m \ge \log^{10}{n}$. Now again by applying Chernoff bound, we conclude that
			\begin{equation}\label{eq:dv-bound}
				\E{d_v} \le d/m + m \in O(m)
			\end{equation}
			holds w.h.p.
			
		\paragraph{Combining the bounds}
			Since \cref{eq:Vi-bound} and \cref{eq:dv-bound} hold independently and w.h.p., by taking union bound over all the vertices we conclude that the number of the edges in $G'[V_i]$ is bounded by $O((n / m) \cdot m)$ w.h.p. This concludes the proof.
	\end{proof}

	\begin{lemma}[Number of phases upper-bound]\label{lemma:number-of-phases}
		$\MPCsimul$ executes $O(\log \log{n})$ phases.
	\end{lemma}
	\begin{proof}
		Let $d_i$ be the degree $d$ of $\MPCsimul$ at the beginning of a phase, and $d_{i + 1}$ the degree updated at \cref{line:update-degree} after the phase ends. Let $I = (\log{m}) / (10 \log{10}) = (\log{d_i}) / (20 \log{10})$. Then, by the definition of the algorithm, we have the following relation
		\begin{equation}\label{eq:di-di+1-relation}
			d_{i + 1} = d_i (1 - \eps)^I = d_i \rb{\frac 12}^{\log{(1 / (1 - \eps)} \frac{\log{d_i}}{20 \log{10}}} = d_i^{1 - \frac{\log{(1 / (1 - \eps)}}{20 \log{10}}}.
		\end{equation}
		For the sake of brevity, define $\gamma := \tfrac{\log{(1 / (1 - \eps))}}{20 \log{10}}$. Observe that for a constant $\eps$ such that $0 < \eps < 1/2$ it implies that $\gamma$ is a constant and $\gamma < 1$. Then, from \cref{eq:di-di+1-relation} we have
		\[
			d_i = d_0^{(1 - \gamma)^i} = n^{(1 - \gamma)^i}.
		\]
		$\MPCsimul$ is executed for $\istar$ phases, where $\istar$ is the smallest integer such that $d_{\istar} \le \log^{20}{n}$. This implies
		\[
			n^{(1 - \gamma)^{\istar}} \le \log^{20}{n}.
		\]
		Taking $\log$ on the both sides of the last inequality, we obtain
		\[
			(1 - \gamma)^{\istar} \log{n} \le 20 \log \log{n}.
		\]
		Now a simple calculation shows that $\istar \in O\rb{\tfrac{\log \log{n}}{\log{(1 / (1 - \gamma))}}} \in O(\log \log{n})$.	
	\end{proof}
	
	\begin{proof}[Proof of \cref{lemma:memory-and-round-complexity}]
		By \cref{lemma:number-of-phases}, $\MPCsimul$ executes $O(\log \log{n})$ phases. Also, for constant $\eps$, \cref{line:directly-simulate} requires $O(\log \log{n})$ iterations. Furthermore, by \cref{lemma:size-of-induced-graphs}, each induced graph $G'[V_i]$ that is processed on a single machine during a phase has $O(n)$ size w.h.p.
		
		There is an implementation of $\MPCsimul$ such that, when every machine has space $O(n)$, each phase and each of the operations the algorithm performs are executed in $O(1)$ MPC-rounds. For more details on such implementation, we refer the reader to~\cite{czumaj2017round}, section \emph{MPC Implementation Details}, and to~\cite{goodrich2011sorting}.
	\end{proof}

\subsubsection{Properties of vertex- and edge-weights in $\MPCsimul$}
\label{sec:weight-properties}
In this section, we show that $\abs{\yideal_v - \ylocal_v}$ remains small for most of the vertices (this claim is formalized by \cref{lemma:max-diff-over-phase}). Before we provide an outline of the analysis, we state some definition and describe the notation we use.

\begin{definition}[Bad and good vertex]
\label{definition:bad-vertex}
	We say that vertex is \emph{bad} in a given phase if it gets frozen in $\LocalRand$ and not in $\MPCsimul$, or the other way around. Once a vertex becomes bad, it remains labeled bad throughout the whole phase, even if it becomes frozen in both $\LocalRand$ and in $\MPCsimul$. If a vertex is not bad, we say it is \emph{good}. In the beginning of a phase, all vertices are initialized as good.
\end{definition}

\begin{definition}[Local neighbor]
If a vertex $u \in N(v)$ is on the same machine as $v$ in the given iteration of $\MPCsimul$, then we say that $u$ is a \emph{local} neighbor of $v$.
\end{definition}

\paragraph{Notation}
We use $w_t$ to refer to the weight of active edges in the beginning of the $t^{th}$ iteration. Let $\NAideal(v, t)$ (resp. $\NAlocal(v, t)$) denote the active neighbors of $v$ at the beginning of the $t^{th}$ iteration of the ideal (resp. MPC) algorithm. Similarly, we use $\Nlocal(v, t)$ to denote the local neighbors of $v$ in iteration $t$. If it is clear from the context which iteration we are referring to, sometimes we omit $t$ from the notation. Throughout our proofs, we will be making claims of the following form $a = b \pm c$, which should be read as $a \in [b - c, b + c]$.

\paragraph{Analysis Outline} Recall that $\ylocal_{v, t}$ and $\yideal_{v, t}$ represent the fractional weight of vertex $v$ in the $t^{th}$ iteration of $\MPCsimul$ and $\LocalRand$, respectively. From our definition, we have $\yideal_{v, t} = \yideal_{v, t - 1} + \eps\, w_{t}\, |\NAideal(v, t)|$, and similarly $\ylocal_{v, t} = \ylocal_{v, t - 1} + \eps \, w_{t} \, \machines \, |\NAlocal(v, t)|$.
To say that the algorithms stay close to each other, we upper-bound $\abs{\yideal_{v, t} - \ylocal_{v, t}}$ inductively as a function of $t$. Suppose that we already have an upper bound on $\abs{\yideal_{v, t - 1} - \ylocal_{v, t - 1}}$\footnote{In our analysis, we assume an upper-bound on a somewhat different quantity $\difflocal(v, t-1)$ given in \cref{definition:difflocal}. It can be shown that $\abs{\yideal_{v, t - 1} - \ylocal_{v, t - 1}} \le \difflocal(v, t-1)$, see \cref{observation:difflocal-upper-bound-y-diff}.}; we focus on upper-bounding the difference between $|\NAideal(v, t)|$ and $\machines |\NAlocal(v, t)|$. There are two parts of the algorithm that affect that difference:
\begin{itemize}
	\item[(1)] Neighbors of $v$ might be bad. Moreover, due to bad vertices, $\NAlocal(v, t)$ might not even be a subset of $\NAideal(v, t)$.
	
	\item[(2)] Even in the very first iteration, or more generally even if there is no bad vertex in $\NAlocal(v, t)$, the set $\NAlocal(v, t)$ is \emph{a random} sample of $\NAideal(v, t)$. Hence, $|\NAlocal(v, t)|$ deviates from its expectation $|\NAideal(v, t)| / \machines$, contributing to the mentioned difference.
\end{itemize}
In our analysis, we assume that at the beginning of each phase $\MPCsimul$ and $\LocalRand$ start from the same fractional matching. Namely, we compare $\MPCsimul$ to the behavior of $\LocalRand$ letting the initial $x$ equal $\xMPC$, for the value of $\xMPC$ at the beginning of a given phase. Since we ensure that $\xMPC$ is at the beginning of a phase always a valid fractional matching, $\LocalRand$ in our approach will also maintain a valid fractional matching.

Also, we assume that the thresholds, i.e., $\cT_{v, t}$ for each $v \in V$ and each iteration $t$, are the same for both $\MPCsimul$ and $\LocalRand$. Note that the latter algorithm is only a hypothetical one, whose purpose is to compare our simulation to a process that constructs a fractional matching, so this assumption is made without loss of generality.

In the rest of this analysis and for the sake of brevity, we assume that $\eps \le 1/2$.

\paragraph{Analysis}

The following claim is a direct consequence of choosing the thresholds randomly in each iteration.
\begin{lemma}\label{claim:prob-to-differ}
	Consider the $t^{th}$ iteration of a phase. Let $\abs{\yideal_{v, t} - \ylocal_{v, t}} \le \sigma$ for every vertex $v$ that is active in both $\LocalRand$ and $\MPCsimul$. Then, $v$ becomes bad in the $t^{th}$ iteration with probability at most $\eps / \sigma$ and independently of other vertices.
\end{lemma}
\begin{proof}
	If $\abs{\ylocal_{v, t} - \cT_{v, t}} > \sigma$, then $\MPCsimul$ and $\LocalRand$ would behave the same with respect to vertex $v$. Since $\cT_{v, t}$ is chosen uniformly at random within interval of size $2 \eps$, $\MPCsimul$ and $\LocalRand$ would differ in iteration $t$ with respect to $v$ with probability at most $2 \sigma / (2 \eps) = \sigma / \eps$. Furthermore, as $\cT_{v, t}$ is chosen independently of other vertices, $v$ becomes bad independently of other vertices.
\end{proof}

There are two distinct steps where $\MPCsimul$ directly or indirectly estimates $\yideal$. The first one is computing $\ylocal$, which is used to deduce whether a vertex should be frozen or not. The second one corresponds to the actual weight that $\MPCsimul$ assigns to the vertices. Namely, at the end of a phase, weight is assigned to each edge (\cref{line:account-for-distinct-machines}) -- for edge $e = \{u ,v\}$, if $u$ or $v$ is frozen, then it is set $\xMPC_e = w_t$, where $t$ is the iteration when the first of the two vertices got frozen; otherwise, $\xMPC_e = w_t$ for $t$ being the most recent simulated iteration. Then, the weight of a vertex $v$, that we denote by $\yMPC_v$, is simply the sum of all $\xMPC_e$ incident to $v$. This can be seen as an indirect estimate of $\yideal_v$.

Our next goal is to understand how does the estimate $\ylocal_v$ and simulated vertex weight $\yMPC_v$ relate to $\yideal_v$. To that end, we define the notion to capture the difference in how the weights $\yideal_v$ and $\yMPC_v$ are composed.
\begin{definition}[Weight-difference]
	We use $\diff(v, t)$ to denote the total weight of the edges that contributed to the weight of $\yideal_{v, t}$ and not to $\yMPC_{v, t}$, and the other way around. Formally, let $\xMPC_{e, t}$ be the updated weight of edge $e$ in iteration $t$ in $\MPCsimul$ (updated in the sense as given by \cref{line:account-for-distinct-machines} of the algorithm). Let $x_{e, t}$ be the weight of edge $e$ in iteration $t$ in $\LocalRand$. Then, 
	\[
		\diff(v, t) := \sum_{e \in N(v)} \abs{x_{e, t} - \xMPC_{e, t}}.
	\]
\end{definition}
Notice that $\abs{\yideal_{v, t} - \yMPC_{v, t}} \le \diff(v, t)$. In general it might be the case that $\abs{\yideal_{v, t} - \yMPC_{v, t}} < \diff(v, t)$. For instance, consider two edges $e_1$ and $e_2$ both incident to $v$. Assume that in algorithm $\LocalRand$ $e_1$ is active, while $e_2$ is frozen. On the other hand, assume that in $\MPCsimul$ it is the case that $e_1$ is frozen while $e_2$ active. So, these two edges alone do not make any difference in the change of the weight of $\yideal_{v, t}$ and $\yMPC_{v, t}$ -- their effects cancel out. However, their effects do not cancel each other in the definition of $\diff(v, t)$.

Similarly, to track $\abs{\yideal_{v, t} - \ylocal_{v, t}}$ we define the following notion.
\begin{definition}[Weight-difference local]\label{definition:difflocal}
	Let $t$ be an iteration of a phase and $\tstar$ be the very first iteration of the same phase.
	Let $\Blocal_{v, t}$ be the set of bad vertices at the beginning of iteration $t$ that are local neighbors of $v$.\footnote{Recall that once a vertex becomes bad in a given phase it remains bad throughout rest of the phase; see \cref{definition:bad-vertex}.} The set $\Blocal_{v, t}$ accounts for the vertices $\NAlocal(v, \hatt) \setminus \NAideal(v, \hatt)$, and for the vertices in $\Nlocal(v, \hatt) \cap \NAideal(v, \hatt)$ but not in $\NAlocal(v, \hatt)$, for $\hatt = \tstar \ldots t$.
	
	We define
	\begin{align}
		\difflocal(v, t) := & w_{\tstar} \cdot \abs{\machines \cdot |\Nlocal(v, \tstar) \cap \NAideal(v, \tstar)| - |\NAideal(v, \tstar)|} \label{eq:difflocal-tstar} \\
		+ & \sum_{\hatt = \tstar + 1}^{t} \eps \cdot w_{\hatt} \cdot \abs{\machines \cdot |\Nlocal(v, \hatt) \cap \NAideal(v, \hatt)| - |\NAideal(v, \hatt)|} \label{eq:difflocal-sum-iterations} \\
			+ & m \cdot w_t \cdot |\Blocal_{v, t}|. \label{eq:difflocal-bad-vertices}
	\end{align}
\end{definition}
We first show that $\difflocal(v, t)$ is a desired quantity.
\begin{lemma}
\label{observation:difflocal-upper-bound-y-diff}
	It holds
	\[
		\abs{\yideal_{v, t} - \ylocal_{v, t}} \le \difflocal(v, t).
	\]
\end{lemma}
\begin{proof}
	The proof of this claim we gave in an informal way by our discussion above. Here we expand that discussion.
	
	Let $\tstar$ be the very first iteration of the phase that iteration $t$ belongs to.
	We have
	\begin{align}
		\abs{\yideal_{v, t} - \ylocal_{v, t}} = & \abs{\yideal_{v, \tstar} + \sum_{\hatt = \tstar + 1}^t \rb{\yideal_{v, \hatt} - \yideal_{v, \hatt - 1}} - \rb{\ylocal_{v, \tstar} + \sum_{\hatt = \tstar + 1}^t \rb{\ylocal_{v, \hatt} - \ylocal_{v, \hatt - 1}}}} \nonumber \\
		& \le \abs{\yideal_{v, \tstar} - \ylocal_{v, \tstar}} + \sum_{\hatt = \tstar + 1}^t \abs{\rb{\yideal_{v, \hatt} - \yideal_{v, \hatt - 1}} - \rb{\ylocal_{v, \hatt} - \ylocal_{v, \hatt - 1}}}. \label{eq:difference-ylocal-yideal-rewritten}
	\end{align}
	\textbf{First}, consider the difference between $\yideal_{v, t}$ and $\ylocal_{v, t}$ coming from random partitioning and \emph{assume that no vertex is bad} -- we will account for bad vertices afterward.
	In \cref{eq:difference-ylocal-yideal-rewritten}, $\abs{\yideal_{v, \tstar} - \ylocal_{v, \tstar}}$ equals the right-hand side of \cref{eq:difflocal-tstar}; there are no bad vertices in iteration $\tstar$ so those two terms are actually equal.
	
	Assuming that no bad vertex exists, then $\ylocal_{v, \hatt} - \ylocal_{v, \hatt - 1}$ equals the weight increase of the vertices in $\Nlocal(v, \hatt) \cap \NAideal(v, \hatt)$ from iteration $\hatt - 1$ to iteration $\hatt$, multiplied by $\machines$. Their weight increases by $w_{\hatt} - w_{\hatt - 1} = \eps w_{\hatt}$. Similarly, $\yideal_{v, \hatt} - \yideal_{v, \hatt - 1} = \eps w_{\hatt} \cdot |\NAideal(v, \hatt)|$. Therefore, \cref{eq:difflocal-sum-iterations} captures the summation in \cref{eq:difference-ylocal-yideal-rewritten}.
\\\\
	\noindent \textbf{Second}, consider bad vertices. Lad $u$ be a bad vertex. Compared to our assumption when no vertex is bad, the vertex $u$ affects $\abs{\yideal_{v, t} - \ylocal_{v, t}}$ in one of two ways:
	$u \in \Nlocal(v, \hatt) \cap \NAideal(v, \hatt)$ but $u \notin \NAlocal(v, \hatt)$; or, $u \in \NAlocal(v, \hatt)$ but $u \notin \NAideal(v, \hatt)$. In the former case, $u$ increases $\yideal_{v, \hatt}$ but not $\ylocal_{v, \hatt}$. In the latter case, $u$ increases $\ylocal_{v, t}$ but not $\yideal_{v, t}$. Nevertheless, compared to our assumption when no vertex is bad, the total effect of $u$ on $\abs{\yideal_{v, t} - \ylocal_{v, t}}$ is upper-bounded by $\machines \cdot \rb{w_{\tstar} + \sum_{\hatt = \tstar + 1}^{t} \eps \cdot w_{\hatt}} = \machines \cdot w_t$. This effect, over all bad vertices, is captured by \cref{eq:difflocal-bad-vertices}.
	
	This concludes the analysis.
\end{proof}

As a first step, we show that $\diff(v, \cdot)$ and $\difflocal(v, \cdot)$ are small in the first iteration of a phase.
\begin{lemma}\label{claim:first-iteration}
	Let iteration $\tstar$ be the first iteration of a phase, and let $v$ be an active vertex by iteration $\tstar$. Then, w.h.p.
	\[
		\difflocal(v, \tstar) \le \machines^{-0.2}.
	\]
	Furthermore, $\diff(v, \tstar) = 0$ with certainty.
\end{lemma}
\begin{proof}
	To argue that $\diff(v, \tstar) = 0$ it suffices to observe that, at the beginning of a phase, edge-weights in $\MPCsimul$ and $\LocalRand$ coincide.

	We upper-bound $\difflocal(v, \tstar)$ as follows. At the beginning of the phase, no vertex is bad. Hence, $\difflocal(v, \tstar)$ only accounts for the random partitioning of the vertices. Also, note that $\NAlocal(v, \tstar) = \Nlocal(v, \tstar) \cap \NAideal(v, \tstar)$.
	
	Define $\mu \eqdef \E{|\NAlocal(v, \tstar)|}$. We now consider two cases, based on $\mu$.
	\paragraph{Case $\mu \le \machines^{0.6}$}
		Then, from Chernoff bound (\cref{lemma:chernoff}~\ref{item:delta-at-least-1}) we have
		\begin{equation}\label{eq:case-mu-small}
			\Prob{\abs{|\NAlocal(v, \tstar)| - \mu} \ge \machines^{0.6}} = \Prob{|\NAlocal(v, \tstar)| - \mu \ge \machines^{0.6}} \le \exp{\rb{-\machines^{0.6} / 3}},
		\end{equation}
		which is high probability as $\machines \ge \log^{10}{n}$ during every phase.
		
	\paragraph{Case $\mu > \machines^{0.6}$}
		Now again by Chernoff bound (\cref{lemma:chernoff}~\ref{item:delta-at-most-1}), we have
		\begin{equation}\label{eq:case-mu-large}
			\Prob{\abs{|\NAlocal(v, \tstar)| - \mu} \ge \machines^{-0.2} \mu} \le 2 \exp{\rb{\machines^{-0.4} \mu / 3}} \le 2 \exp{\rb{\machines^{0.2} \mu / 3}}.
		\end{equation}
		This probability is again high, as $\machines \gg \log{n}$.
	
	\paragraph{Combining the two cases}
	Observe that $\NAlocal(v, \tstar)$ is a random sample of $\NAideal(v, \tstar)$, and hence $\mu = \tfrac{|\NAideal(v, \tstar)|}{\machines}$. From \cref{eq:case-mu-small} and \cref{eq:case-mu-large} we derive that w.h.p.~it holds
	\[
		|\NAlocal(v, \tstar)| = \frac{|\NAideal(v, \tstar)|}{\machines} \pm \max\left\{\machines^{-0.2} \frac{|\NAideal(v, \tstar)|}{\machines}, \machines^{0.6}\right\}.
	\]
	Hence, w.h.p.
	\[
		\difflocal(v, \tstar) \le w_{\tstar} \max\{\machines^{-0.2} |\NAideal(v, \tstar)|, \machines^{1.6}\}.
	\]
	Now, from \cref{obs:degree-active-weight} and \cref{obs:max-active-weight} we have that $|\NAideal(v, \tstar)| w_{\tstar} \le 1$ and $w_{\tstar} \le \machines^{-1.8}$. This implies that w.h.p. it holds
	\[
		\difflocal(v, \tstar) \le \machines^{-0.2},
	\]
	as desired.
\end{proof}

We now prove our main technical lemma, which is used to quantify the increase in the difference between $\yideal_v$ and its estimates $\ylocal_v$ and $\yMPC_v$ over the course of one phase.
\begin{lemma}[Evolution of weight-estimates]\label{claim:change-of-sigma}
	Let $v$ be an active vertex in iteration $t-1$ in both $\LocalRand$ and $\MPCsimul$. Then, if $\difflocal(v, t - 1) \le \sigma$ and $\diff(v, t - 1) \le \sigma$, the following holds w.h.p.:
	\begin{itemize}
		\item $\difflocal(v, t) \le 5 (\sigma + \eps \machines^{-0.2})$, and
		\item $\diff(v, t) \le 5 (\sigma + \eps \machines^{-0.2})$.
	\end{itemize}
\end{lemma}
\begin{proof}
	We proceed by upper-bounding the effect of three different kinds of vertices on $\difflocal(v, t)$ and $\diff(v, t)$: bad vertices prior to the $t^{th}$ iteration; vertices becoming bad in the $t^{th}$ iteration; and, the effect of the random partitioning. Observe that $\diff(v, t)$ is \emph{not directly} affected by the random partitioning.
	
	\paragraph{Old bad vertices and partitioning effect prior to iteration $t$} 
	Let $\Blocal_{v, t - 1}$ be as in \cref{definition:difflocal}. Note that $\sigma$ upper-bounds both the effect of vertex partitioning and the effect of $\Blocal_{v, t - 1}$ before iteration $t$. The weight of the vertices in $\Blocal_{v, t - 1}$ increases by $w_t - w_{t-1} = \tfrac{\eps}{1-\eps} w_{t-1} \le 2\eps w_{t-1}$ from iteration $t-1$ to $t$; recall that we assume $\eps \le 1/2$.
	Hence, effect of vertex partitioning before iteration $t$ and the effect of old bad vertices $\Blocal_{v, t - 1}$ in iteration $t$ is at most $(1 + 2 \eps) \sigma$ combined.
	
	
		In a similar way, define $B_{v, t - 1}$ to be the set of bad neighbors of $v$ across \emph{all} the machines at the beginning of iteration $t - 1$. Then, we get that from iteration $t-1$ to iteration $t$ the effect of the old bad vertices on $\diff(v, t)$ increases by at most $2 \eps \sigma$.
	
	\paragraph{New bad vertices} In addition to the bad vertices in $\Blocal_{v, t - 1}$, there might be new bad vertices in the beginning of the $t^{th}$ iteration -- the vertices of $\NAlocal(v, t - 1) \cap \NAideal(v, t - 1)$ that are not in $\NAlocal(v, t) \cap \NAideal(v, t)$. To upper-bound the weight of those bad vertices, we first upper-bound the cardinality of $\NAlocal(v, t - 1) \cap \NAideal(v, t - 1)$. For the sake of brevity, define
	\[
		\nlocal_{v, t - 1} := |\Nlocal(v, t - 1) \cap \NAideal(v, t - 1)|
	\]
	where, as a reminder, the set $\Nlocal(v, t - 1)$ refers to the local neighbors (both frozen and active) of $v$. We trivially have
		\[
			|\NAlocal(v, t - 1) \cap \NAideal(v, t - 1)| \le \nlocal_{v, t - 1}.
		\]
		Then, by \cref{claim:prob-to-differ,observation:difflocal-upper-bound-y-diff} and our assumption $\difflocal(v, t - 1) \le \sigma$, the number of new bad vertices is in expectation at most $\nlocal_{v, t - 1} \sigma / \eps$. We now proceed by providing a sharp concentration around this expected value. To that end, we provide an upper-bound on $\nlocal_{v, t - 1}$ that holds w.h.p.
\\
Observe that $\NAideal(v, t - 1)$ is defined deterministically and independently of the MPC algorithm. Then, if $|\NAideal(v, t - 1)| \ge \machines^{1.6}$, we have that w.h.p. $\nlocal_{v, t - 1} \le (1 + \machines^{-0.2}) |\NAideal(v, t - 1)| / \machines$. Otherwise, if $|\NAideal(v, t - 1)| < \machines^{1.6}$, then w.h.p. $\nlocal_{v, t - 1} \le 2 \machines^{0.6}$. Therefore, for $\gamma := \max\{(1 + \machines^{-0.2}) |\NAideal(v, t - 1)|, 2 \machines^{1.6}\}$, we have the w.h.p.
\[
	\machines \cdot \nlocal_{v, t - 1} \le \gamma.\footnote{A more detailed proof for this type of claim is given in the proof of \cref{claim:first-iteration}.}
\]
Applying similar reasoning about $\nlocal_{v, t-1} \sigma / \eps$, i.e., considering cases $\nlocal_{v, t-1} \sigma / \eps \ge \machines^{0.6}$ and $\nlocal_{v, t-1} \sigma / \eps < \machines^{0.6}$, we obtain that w.h.p.~the number of new bad vertices is upper-bounded by $\max\{(1 + \machines^{-0.2}) \nlocal_{v, t - 1} \sigma / \eps,  2 \machines^{0.6}\}$. So, putting all together, we have that the weight coming from new bad vertices that affects the local estimate of $\yideal_{v, t}$ is at most 
\[
	\sigma_2 := \eps w_{t} \cdot \max\{(1 + \machines^{-0.2}) \gamma \sigma / \eps, 2 \machines^{1.6}\}.
\]
But now, using that $w_{t} \le \machines^{-1.8}$ and also that $|\NAideal(v, t - 1)| w_{t - 1} \le (1 - 2\eps)$ as $v$ is an active vertex in $\LocalRand$ in iteration $t-1$, we derive $\sigma_2 \le 2 (\sigma +\eps \machines^{-0.2})$.

	It remains to comment about the effect of new bad vertices on $\diff(v, t)$. Note that the expected number of new bad vertices affecting $\diff(v, t)$ is at most $|\NAideal(v, t - 1)| \sigma / \eps$. So, applying the same arguments as above, the weight of new bad vertices affects $\diff(v, t)$ by at most $\sigma_2$ w.h.p.
	
	\paragraph{Effect of random partitioning -- \cref{eq:difflocal-sum-iterations}} Finally, we upper-bound the effect of the random partitioning in iteration $t$ on $\difflocal(v, t)$. Similarly to our arguments given in \cref{claim:first-iteration}, we have that w.h.p.~the number of vertices of $\NAideal(v, t)$ that are local neighbors of $v$ deviates from $|\NAideal(v, t)| / \machines$ by at most $\eta$ defined as
	\[
		\eta := \max\{\machines^{-0.2} |\NAideal(v, t)|, \machines^{1.6}\} / \machines.
	\]
	The total weight of these vertices scaled by $\machines$ is at most $\eps \cdot w_{t} \cdot \machines \cdot \eta \le \eps \machines^{-0.2}$.
	
	\paragraph{Final step} Putting altogether, if
	\[
		\difflocal(v, t - 1) \le \sigma
	\]
	and
	\[
		\diff(v, t - 1) \le \sigma,
	\]
	then we have
	\begin{eqnarray*}
		\difflocal(v, t - 1) & \le & (1 + 2 \eps) \sigma + 2 (\sigma + \eps \machines^{-0.2}) + \eps \machines^{-0.2} \\
		& \le & 5 (\sigma + \eps \machines^{-0.2}),
	\end{eqnarray*}
	and similarly
	\[
		\diff(v, t) \le (1 + 2 \eps) \sigma + 2 (\sigma + \eps \machines^{-0.2}) \le 5 (\sigma + \eps \machines^{-0.2}),
	\]
	as desired.
\end{proof}

Now, combining \cref{observation:difflocal-upper-bound-y-diff,claim:first-iteration,claim:change-of-sigma}, it is not hard to show the following.
\begin{lemma}\label{lemma:max-diff-over-phase}
	Let $v$ be a vertex active in iteration $t-1$ in both $\MPCsimul$ and $\LocalRand$. If a phase consists of at most $I := (\log{m}) / (10 \log 10)$ iterations, then it holds $\abs{\yideal_{v, t} - \ylocal_{v, t}} \le \machines^{-0.1}$ and $\diff(v, t) \le \machines^{-0.1}$ w.h.p.
\end{lemma}
\begin{proof}
	Let iteration $\tstar$ be the first iteration of the $i^{th}$ phase. Combining \cref{observation:difflocal-upper-bound-y-diff}, \cref{claim:first-iteration} and \cref{claim:change-of-sigma}, for any $\tstar \le t \le \tstar + I$ in which $v$ is not bad, it holds
	\[
		\abs{\yideal_{v, t} - \ylocal_{v, t}} \le \difflocal(v, t) \le 10^I \machines^{-0.2} \le \machines^{-0.1},
	\]
	and
	\[
		\diff(v, t) \le 10^I \machines^{-0.2} \le \machines^{-0.1}.
	\]
\end{proof}

We are now ready to prove the main result of this section.
\subsubsection{Proof of \cref{lemma:main-lemma-matching}}
\label{sec:proof-of-matching-lemma}
	We start the proof by recalling that \cref{lemma:memory-and-round-complexity} shows the desired bound on the space- and round-complexity of $\MPCsimul$. The rest of the proof is divided into two parts. First, we prove the statement for vertex cover, and then for matching.

	Throughout the proof we consider only those rounds of the MPC algorithm that execute at least two iterations. The rounds in which is executed only one iteration coincide with the ideal algorithm, and for them the claims in the rest of the proof follow directly. In this section, we assume that a maximum matching and a minimum vertex cover is of size at least $\log^{10}{n}$. In \cref{sec:removing-matching-assumption} we show how to handle that case when the maximum matching has size less than $\log^{10}{n}$.

	\paragraph{Part I --- Vertex Cover}

		Let $\tC$ be the vertex cover constructed by $\MPCsimul$. First, observe that $\tC$ is indeed a vertex cover as by the end of the algorithm every edge is incident to at least one frozen vertex, and every frozen vertex is included in $\tC$. This follows as: by the end of the algorithm the weight of active edges is at least $1 - 2\eps$; and, the last iterations of $\MPCsimul$ directly simulate $\LocalRand$. Since $\LocalRand$ freezes any vertex (and its incident edges) having incident edge of weight at least $1 - 2\eps$, $\MPCsimul$ freezes such vertices as well.
		
		Informally, our goal is to show that $|\tC|$ is roughly at most twice larger than $W_M := \sum_{v \in V'} \yMPC_v$, where $V'$ is the set of vertices after removing those of weight more than $1$. Our proof consists of two main parts. First, we consider the contribution to $W_M$ of the vertices that remained active in $\LocalRand$ for at least as many iterations in $\MPCsimul$ (and at first we ignore the other vertices). After that, we take into account the remaining vertices, and in the same time account for the vertices having weight more than $1$.
		
	\paragraph{$\LocalRand$ freezing last}
		Let $t$ be the last iteration of a phase. We first consider only those vertices added to $\tC$ which remained active in $\LocalRand$ for at least as many iterations as in $\MPCsimul$, and claim that for every such vertex $v$ it holds $\yMPC_v \ge 1 - 5 \eps$. In the analysis we give, we ignore that some vertices $u$ such that $\yMPC_u > 1$ got removed along with their incident edges. We analyze two types of vertices: good vertices; and, bad vertices that got frozen by $\MPCsimul$ first.
		\begin{itemize}
			\item[(1)] If $v$ is good, then it was active in $\MPCsimul$ in the same iterations as in $\LocalRand$. Hence, by \cref{lemma:max-diff-over-phase}, we have that $\yMPC_{v, t} \ge \yideal_{v, t} - \machines^{-0.1} \ge 1 - 4 \eps - \eps = 1 - 5\eps$.
			\item[(2)] Assume that $v$ is bad, but got frozen by $\MPCsimul$ first. Let $t'$ be the iteration $v$ got frozen by $\MPCsimul$. This directly implies that $\ylocal_{v, t} \ge 1 - 4\eps$. Since $v$ was active in the both algorithms in iteration $t' - 1$, by \cref{lemma:max-diff-over-phase} we have $\yideal_{v, t'} \ge 1 - 4 \eps - \machines^{-0.1}$. But now again by \cref{lemma:max-diff-over-phase} we conclude that $\yMPC_{v, t'} \ge 1 - 4\eps - \machines^{-0.1} - \machines^{-0.1} \ge 1 - 5 \eps$.
		\end{itemize}
		Informally (again), this analysis can be stated as: for every vertex of the two considered types which is added to $\tC$, and while disregarding the vertices whose incident edges got removed, there is at least $(1 - 5 \eps) / 2$ weight in $W_M$. The weight is scaled by $2$ as every edge is incident to at most $2$ vertices of $\tC$.

	\paragraph{$\MPCsimul$ freezing last}
		We now consider the vertices that got frozen by $\MPCsimul$ in later iteration that by $\LocalRand$. We call such vertices \emph{late-bad}, and use $\nlate$ to denote their number. Let $C$ denote the vertex cover constructed by $\LocalRand$. Observe that the late-bad vertices are a subset of $C$ --- if vertex is not active in $\LocalRand$ anymore, it means it has been frozen and added to $C$. Late-bad have another important property -- every vertex $v$ such that $\yMPC_v > 1$ is late-bad, as we argue in the sequel. In our next step, we upper-bound $\nlate$ by the number of the vertices of $C$ that are bad.
		
		Let $C_t$ denote the vertices that join the vertex cover $C$ in the $t^{th}$ iteration of $\LocalRand$. From \cref{claim:prob-to-differ} and \cref{lemma:max-diff-over-phase}, a vertex is bad with probability at most $\machines^{-0.1} / \eps$. Hence, the expected number of bad vertices in the $t^{th}$ iteration is at most $\machines^{-0.1} |C_t| / \eps$. Notice that $C_t$ is a deterministic set, defined independently of $\MPCsimul$. Furthermore, at the $t^{th}$ iteration, every vertex of $C_t$ becomes bad independently of other vertices. So, the number of bad, and also heavy-bad, vertices throughout all the phases is with high probability upper-bounded by $O(\max\{\log^2{n}, \machines^{-0.1} |C| / \eps\})$. Recall that we assume $\machines^{-0.1} \le \eps^2$, and also that a minimum vertex cover of the graph has size at least $\log^{10}{n}$. This now implies
	\[
		|\tC| \ge \rb{1 - \frac{\machines^{-0.1}}{\eps}} |C| \ge (1 - \eps) |C|,
	\]
	and hence
	\begin{equation}\label{eq:nlate-upper-bound}
		\nlate \le \eps |C| \le \frac{\eps}{1 - \eps} |\tC|.
	\end{equation}

	\paragraph{Vertices $v$ such that $\yMPC_v > 1$}
		We say that a vertex $v$ is \emph{heavy-bad} if $\yMPC_v > 1$. The analysis we performed above on relating $|\tC|$ and $W_M$ does not take into account heavy-bad vertices. Recall that heavy-bad vertices are removed from the graph along with their incident edges, which we did not account for while lower-bounding $\yMPC_u$ for $u \in \tC$. Next, we discuss how much heavy-bad vertices affect $\yMPC_u$ for any vertex $u \in V' \cap \tC$.
		
		Observe that $v$ is heavy-bad only if: $v$ belongs to some set $C_t$; $v$ was active in the $(t - 1)^{st}$ iteration by $\MPCsimul$; and, $v$ remained active (by $\MPCsimul$) throughout the $t^{th}$ iteration. Hence, every heavy-bad vertex is also late-bad (but there can be a late-bad vertex that is not heavy-bad).
		
	Let $v$ be late-bad. Observe that by the time it holds $\yMPC_v > 1$, vertex $v$ is already bad w.h.p. --- as long as $v$ is not bad from \cref{lemma:max-diff-over-phase} we have $\yMPC_v \le \yideal_v + \machines^{-0.1} \le 1 - \eps + \machines^{-0.1} < 1$. On the other hand, from the iteration $v$ got frozen in $\LocalRand$, along with its incident edges, the increase in $\yMPC_v$ is accounted to $\diff(\cdot, \cdot)$, which we have already analyzed. So, to account for the removal of heavy-bad vertices and their incident edges, it suffices to upper-bound the total weight of heavy-bad vertices while they were still active in $\LocalRand$. That weight is trivially upper-bounded by $\nlate$. Furthermore, the weight $\nlate$ takes into account those vertices that are late-bad but not heavy-bad, so we do not consider separately such vertices (as we did for the other kind of bad vertices and for the good ones).
	
	\paragraph{Finalizing}
		Let $\tClate$ be the subset of $\tC$ consists of late-bad vertices. Our analysis shows
		\[
			\frac{W_M + \nlate}{|\tC| - |\tClate|} \ge \frac{1 - 5 \eps}{2}.
		\]
		For the sake of brevity, define $\alpha := (1 - 5 \eps) / 2$. Using that $|\tClate| \le \nlate$ and upper-bound \cref{eq:nlate-upper-bound}, we derive
		\begin{eqnarray}
			W_M & \ge & \alpha |\tC| - \nlate (1 + \alpha) \nonumber \\
				&\ge & \rb{\alpha - \frac{\eps}{1 - \eps} (1 + \alpha)} |\tC| \label{eq:cT-bound}
		\end{eqnarray}
		Next, observe that $\alpha < 1$ and hence $1 + \alpha < 2$. Also, we assume $\eps < 1/2$. Then, \eqref{eq:cT-bound} further implies
		\[
			W_M \ge \rb{\alpha - 4 \eps} |\tC|,
		\]
		and hence
		\begin{equation}\label{eq:tC-upper-bound}
			|\tC| \le \frac{2}{1 - 13 \eps} W_M \le 2 (1 + 50 \eps) W_M.
		\end{equation}
		Finally, from strong duality it implies that $|\tC| \le 2 (1 + 50 \eps) W_C^{\star}$, where $W_C^{\star}$ is the minimum fractional vertex cover weight. Since the minimum integral vertex cover has size at least $W_C^{\star}$, the lemma follows.

	\paragraph{Part II --- Maximum Matching}	
				
	After we provided an upper-bound for $\tC$ by \cref{eq:tC-upper-bound}, the analysis of the weight of fractional maximum matching our algorithm $\MPCsimul$ designs follows almost directly. First, recall that $W_M$ denotes the weight of the fractional matching $\MPCsimul$ designs. Also recall that $W_M$ does not include the vertices $v$ that got removed due to having $\yMPC_v > 1$. Therefore, by the design of the algorithm, the vertex-weights $\yMPC_v$ satisfy the matching constraint, i.e., $\yMPC_v \le 1$. Furthermore, from \cref{eq:tC-upper-bound} and from the fact $|\tC| \ge W_C^{\star}$ we have
	\[
		W_M \ge \frac{1}{2 (1 + 50 \eps)} |\tC| \ge \frac{1}{2 (1 + 50 \eps)} W_C^{\star}.
	\]
	Now by strong duality, $W_M$ is a $(2 (1 + 50 \eps))$-approximation of fractional maximum matching.

\subsubsection{Finding small matchings and vertex covers}
\label{sec:removing-matching-assumption}
In the proof of \cref{sec:proof-of-matching-lemma} we made an assumption that the maximum matching size is at least $\log^{10}{n}$. If the maximum matching size is less than $\log^{10}{n}$, in this section we show how to find a maximal matching and a $2$-approximate minimum vertex cover in $O(\log \log {n})$ rounds when the memory per machine is $\Theta(n)$.

First, observe that if the size of a minimum vertex cover is $O(\log^{10} n)$, then the underlying graph has $O(n \log^{10}{n})$ edges -- each vertex can cover at most $n$ edges.
If our graph has $O(n \log^{10}{n})$ edges, we apply the result of~\cite{LattanziMSV11} to find a maximal matching of the graph in $O(\log \log{n})$ MPC rounds. Namely, in~\cite{LattanziMSV11} in the proof of Lemma~3.2 it is shown that their algorithm w.h.p.~halves the number of the edges in each MPC round. Hence, after $O(\log \log{n})$ the algorithm will produce some matching, and the induced graph on the unmatched vertices will have $O(n)$ edges. After that, we gather all the edges on one machine and find the remaining matching. The endpoints of this maximal matching give a $2$-approximate vertex cover. We point out that it is crucial that their method outputs a maximal matching, so it is easy to turn it into a $2$-approximate minimum vertex cover.


\section{Integral Matching and Improved Approximation}
\label{sec:randomized-rounding}
	In this section we prove the following theorem.
	\theoremMatchingVC*
	
	Before we provide a proof, recall that \cref{lemma:main-lemma-matching} shows how to construct \emph{a fractional} matching of large size. In the following lemma we show how to round that matching (i.e., to obtain an integral one), while still retaining large size of the fractional matching. This lemma is the main ingredient of the proof of \cref{thm:matching-VC}.
	\begin{lemma}[Randomized rounding]\label{lemma:randomized-rounding}
		Let $G = (V, E)$ be a graph. Let $x\ :\ E \to [0, 1]$ be a fractional matching of $G$, i.e., for each $v \in V$ it holds $\sum_{e \ni v} x_e \le 1$. Let $\tC \subseteq V$ be a set of vertices such that for each $v \in \tC$ it holds $\sum_{e \ni v} x_e \ge 1 - \beta$, for some constant $\beta \le 1/2$. Then, there exists an algorithm that with probability at least $1 - 2 \exp{\rb{-|\tC|/5000}}$ outputs a matching in $G$ of size at least $|\tC| / 50$.
	\end{lemma}
	In our proof of this lemma we use McDiarmid's inequality that we review first.
	\begin{theorem}[McDiarmid's inequality]\label{theorem:McDiarmid}
		Suppose that $X_1, \ldots, X_k$ are independent random variables and assume that $f$ is a function that satisfies
		\[
			\sup_{x_1, \ldots, x_k, tx_i} |f(x_1, \ldots, x_k) - f(x_1, \ldots, x_{i - 1}, \tx_i, x_{i + 1}, \ldots, x_k)| \le c, \text{ for all } 1 \le i \le k.
		\]
		(The inequality above states that if one coordinate of the function is changed, then the value of the function changes by at most $c$.)
		
		Then, for any $\delta > 0$ it holds
		\[
			\Prob{|f(X_1, \ldots, X_k) - \E{f(X_1, \ldots, X_k)}| \ge \delta} \le 2 \exp{\rb{-\frac{2 \delta^2}{k c^2}}}.
		\]
	\end{theorem}
	
	\begin{proof}[Proof of \cref{lemma:randomized-rounding}]
		Our goal is to apply \cref{theorem:McDiarmid} in order to prove this lemma. So we will design a randomized process that will correspond to the setup of the theorem, but also round the fractional matching $x$.
		
		\paragraph{Setup and the rounding algorithm}
			For every vertex $v \in \tC$ we define a random variable $X_v$ as follows. $X_v$ takes value from the set $\{N(v) \cup \{\spec\}\}$. So, $X_v$ is either a neighbor of $v$ or a special symbol $\spec$. Intuitively, $X_v$ will correspond to $v$ (randomly) choosing some neighbor, and if $X_v$ equals $\spec$, then it would mean $v$ have not chosen any of the neighbors. The probability space for each $X_v$ is defined as follows: for every $u \in N(v)$, we define $\Prob{X_v = u} = x_{\{u, v\}} / 10$, and $\Prob{X_v = \spec} = 1 - (\sum_{e \ni v} x_e) / 10$. Observe that $\Prob{X_v = \spec} \ge 9/10$. For any two vertices $u, v \in \tC$, the random variables $X_u$ and $X_v$ are chosen independently.
		
		Now we define function $f$. First, given a set of edges $H$ we say that edge $e \in H$ is \emph{good} if $H \setminus \{e\}$ does not contain edge incident to $e$. For a set of variables $\{X_v\}_{v \in \tC}$ we construct a set of edges $H_X$ as follows: if $X_v \neq \spec$ we add edge $\{v, X_v\}$ to $H_X$; otherwise $X_v$ does not contribute to $H_X$. Let $\{v_1, \ldots, v_{|\tC|}\}$ be the vertices of $\tC$. We set $f(X_{v_1}, \ldots, X_{v_{|\tC|}})$ to be the number of good edges in $H_X$.
		
		The number of good edges obtained in this random process represent our rounded matching. Next, we lower-bound the size of the integral matching obtained in this way. To that end, we derive the upper-bound on $c$ for $f$ as defined in \cref{theorem:McDiarmid} and lower-bound the expectation of $f$.
		
		\paragraph{Upper-bound on $c$}
			Fix a vertex $v \in \tC$. If $X_v = \spec$, then $X_v$ does not contribute any edge to $H_X$. If $X_v$ would change to some neighbor of $v$, then it would result in adding edge $e = \{X_v, v\}$ to $H_X$. But now, if there were good edges incident to $v$ or $X_v$, they will not be good anymore. So, changing $X_v$ from $\spec$ to a neighbor of $v$ could increase $f$ by at most $2$. On the other hand, if there was no edge incident to $\{X_v, v\}$ in $H_X$, then changing $X_v$ in the described way would increase $f$ by $1$.
			
				Assume now that $X_v \neq \spec$. Then, similarly to the analysis above, changing $X_v = u$ to another neighbor $u'$ of $v$ could increase $f$ by $2$ at most if $u$ initially had two incident edges while $u'$ had none, so by changing $X_v$ to $u'$ there are two more good edges. In the opposite way, the number of good edges could be decreased by $2$ at most. Finally, changing $X_v$ to $\spec$ could increase $f$ by at most $2$ or decrease by at most $1$.
				
				From this case analysis, we conclude $c = 2$.
				
		\paragraph{Lower-bounding the expectation of $f$}
			Consider an edge $e = \{u, v\}$ incident to a vertex $v \in \tC$. Now we will analyze when $\{X_v = u, v\}$ is a good edge. If $X_u = \spec$, and for every neighbor $w \in N(v) \cap \tC$ we have $X_w \neq v$, the variable $X_v = u$ will contribute $1$ to $f$. First, $\Prob{X_u = \spec} \ge 9 / 10$. On the other hand
			\begin{equation}\label{eq:X_w-for-all-bound}
				\Prob{X_w \neq v \text{ for all } w \in N(v) \cap \tC} = \prod_{e \ni v} \rb{1 - \frac{x_e}{10}} \ge \exp{\rb{-\sum_{e \ni v} \frac{x_e}{10} - \sum_{e \ni v} \frac{x_e^2}{100}}},
			\end{equation}
			where we used inequality $- \ln{(1 - y)} \le y + y^2$ that holds for $|y| \le 1/2$. Now using $y^2 \le y$ for $0 \le y < 1$ and $\sum_{e \ni v} x_e \le 1$, from \cref{eq:X_w-for-all-bound} we further have
			\[
				\Prob{X_w \neq v \text{ for all } w \in N(v) \cap \tC} \ge \exp{\rb{- \frac{11}{100} \sum_{e \ni v} x_e}} \ge \exp{\rb{- \frac{11}{100}}} \ge \frac{89}{100},
			\]
			where the last inequality follows from $1 - y \le \exp{\rb{-y}}$.
			
			So, $X_v = u$ contributes $1$ to $f$ with probability at least $\tfrac{x_{\{v, u\}}}{10} \cdot 9/10 \cdot 89/100 \ge 4 x_{\{v, u\}} / 5$. Since for each vertex $v \in \tC$ it holds $\sum_{e \ni v} \tfrac{x_e}{10} \ge \tfrac{1 - \beta}{10}$, and $\beta \le 1/2$, from linearity of expectation we get
			\begin{equation}\label{eq:exp-f-bound}
				\E{f(X_{v_1}, \ldots, X_{v_{|\tC|}})} \ge 4 |\tC| (1 - \beta) / 50 \ge |\tC| / 25.
			\end{equation}
			
		\paragraph{Applying \cref{theorem:McDiarmid}}
			We are now ready to conclude the proof. Let $\delta = |\tC| / 50$. By applying \cref{theorem:McDiarmid} to the function $f$ and the random variables we defined, using that $c=2$ and the lower-bound \cref{eq:exp-f-bound} on the expectation of $f$, we conclude that $f(X_{v_1}, \ldots, X_{v_{|\tC|}}) \ge |\tC| / 50$ with probability at least $1 - 2 \exp{\rb{-|\tC|/5000}}$.
	\end{proof}
	We are now ready to prove the main theorem.
	
	\theoremMatchingVC*
	\begin{proof}
		Invoking \cref{lemma:main-lemma-matching} for the approximation parameter $\eps / 50$ we obtain the desired approximation of the minimum vertex cover. To obtain a $(2+\eps)$-approximate (integral) maximum matching, we alternatively apply the results of \cref{lemma:main-lemma-matching} and \cref{lemma:randomized-rounding}, as we describe in the sequel. We proceed with the proof as follows. First, we describe how to handle the case when the input graph has small matching. Second, we define an algorithm that iteratively extracts matching of a constant size from our graph. Finally, we analyze the designed algorithm -- the probability of success and the number of required iteration to produce a $(2 + \eps)$-approximate maximum matching.
		
		\paragraph{Small degree} We invoke two methods separately, each of them providing a matching, and we output the larger of them as the final result. The first method is described \cref{sec:removing-matching-assumption}, and performs well when the matching size if $O(\log^{10}{n})$. Hence, from now on we assume that the maximum matching is of size at least $\log^{10}{n}$.
		
		\paragraph{Algorithm}
			Now, define algorithm $\cA$ that as input gets a graph $G = (V, E)$ and consists of the following steps:
		\begin{itemize}
			\item Invoke $\MPCsimul$ to obtain a fractional matching $x$.
			\item Apply the rounding method described by \cref{lemma:randomized-rounding} on $x$. Let $M$ be the produced integral matching.
			\item Update $V$ by removing from it all the vertices in $M$.
		\end{itemize}
	
		\paragraph{Analysis of the algorithm}
			Consider one execution of $\cA$. Let $x$ be the fractional matching returned by $\MPCsimul$ for the approximation parameter set to $\eps / 50$, and let $W(x)$ denote its weight. Let $C$ be the vertex cover as defined in the statement of \cref{lemma:main-lemma-matching}. By \cref{lemma:main-lemma-matching}, and from the fact that $W(x) \le |C|$, there are at least $W(x) / 3$ vertices that have fractional weight at least $1 - 5 \eps$. Hence, as long as $x$ has weight at least $\log^{9}{n}$, the rounding method described by \cref{lemma:randomized-rounding} w.h.p.~produces an integral matching $M$ of size at least $W(x) / 150$.
			
			Consider now multiples executions of $\cA$. Once it holds $W(x) < \log^{9}{n}$, it means that we have already collected a large fraction of any maximal matching, i.e., $(1 - 1/\log{n})$ fraction. On the other hand, as long as $W(x) \ge \log^{9}{n}$ algorithm $\cA$ will produce an integral matching of size at least $1/150$ fraction of the size of the current maximum matching. This discussion motivates our final algorithm which is as follows: run $\cA$ for $\log_{150/149}{(1/ \eps)}$ many iterations and output the union of integral matching it produces. Our discussion implies that the final returned matching is a $(2+\eps)$-approximate maximum matching of the input graph. Furthermore, for constant $\eps$, this algorithm can be implemented in $O(\log \log n)$ MPC-rounds.
	\end{proof}
	
\subsection*{Acknowledgments}
We thank Zeyong Li, Daan Nilis, and anonymous reviewers for their valuable feedback. S.M.~is grateful to his co-authors for the previous collaboration in~\cite{czumaj2017round} that was the starting point of this project. We are also grateful to Christoph Grunau for valuable discussions and for pointing out an imprecision in the previous proof of \cref{claim:change-of-sigma}.
R.R.~was supported by NSF award numbers CCF-1650733, CCF-1733808, CCF-1740751, IIS-1741137 and Israel Science Foundation Grant 1147/09.
Most of the work on this paper has been carried out while C.K.~was at the University of Warwick, where he was supported by the Centre for Discrete Mathematics and its Applications (DIMAP) and by EPSRC award EP/N011163/1. Part of this work has been carried out while S.M.~was visiting MIT.

\bibliographystyle{alpha}
\bibliography{ref}
\end{document}